\documentclass[12pt]{article}

\usepackage[english]{babel}
\usepackage[utf8x]{inputenc}
\usepackage[ruled, vlined]{algorithm2e}
\usepackage{algorithmic}
\usepackage{setspace}
\usepackage{amsmath}
\usepackage{comment}
\usepackage{amsthm, amssymb}
\usepackage{amsfonts}
\usepackage{enumerate}
\usepackage{graphicx}
\usepackage[dvipsnames]{color}
\usepackage{fullpage}
\usepackage{bbm}
\usepackage{xspace}
\usepackage[margin=0.75in]{geometry}
\usepackage{latexsym}
\usepackage[titletoc]{appendix}
\usepackage{tikz}
\usetikzlibrary{positioning, calc, fit, shapes, backgrounds}

\allowdisplaybreaks
\newcommand{\err}{err}
\newtheorem{theorem}{Theorem}

\newtheorem{lemma}[theorem]{Lemma}

\newtheorem{corollary}[theorem]{Corollary}
\newtheorem{claim}[theorem]{Claim}

\newtheorem{observation}[theorem]{Observation}
\newtheorem{remark}[theorem]{Remark}
\newtheorem{proposition}[theorem]{Proposition}
\newtheorem{property}[theorem]{Property}

\theoremstyle{definition}
\newtheorem{definition}[theorem]{Definition}
\newcommand{\ind}{\mathbbm{1}}

\DeclareMathOperator*{\E}{\mathbb{E}}

\newcommand{\IGNORE}[1]{}
\newcommand{\poly}{\textsf{Poly}}

\newcommand{\dist}{\mathcal{D}}

\newcommand{\val}{\ensuremath{\mathsf{Val}}}
\newcommand{\rev}{\ensuremath{\mathsf{{Mye}}}}

\newcommand{\spc}{\text{ }}

\newcommand{\phaseVar}{\tau}

\newcommand{\XVCG}{X_{VCG}}
\newcommand{\PVCG}{P_{VCG}}
\newcommand{\SingleLP}{Single LP}
\newcommand{\BorderLP}{Border LP}

\newcommand{\ReducedUniformLP}{Reduced Uniform LP}
\newcommand{\SingleRelaxLP}{LRSingle}
\newcommand{\BorderRelaxLP}{LRBorder}

\newcommand{\SingleOPT}{SLPRev}
\newcommand{\BorderOPT}{BLPRev}
\newcommand{\xst}{x^*}
\newcommand{\yst}{y^*}
\newcommand{\ust}{u^*}

\newcommand{\D}{\mathcal{D}}

\begin{document}
\title{Selling to Multiple No-Regret Buyers}
\author{Linda Cai\footnote{Computer Science, Princeton University. \texttt{tcai@princeton.edu}.} \and S. Matthew Weinberg\footnote{Computer Science, Princeton University. Supported by NSF CAREER CCF-1942497 and a Google Research Award.   \texttt{smweinberg@princeton.edu}.} \and Evan Wildenhain\footnote{Google. \texttt{m@evanwildenha.in}.} \and Shirley Zhang\footnote{Harvard University. \texttt{skzhang@alumni.princeton.edu}}.}
\maketitle
\begin{abstract}
We consider the problem of repeatedly auctioning a single item to multiple {i.i.d} buyers who each use a no-regret learning algorithm to bid over time. In particular, we study the seller's optimal revenue, if they know that the buyers are no-regret learners (but only that their behavior satisfies some no-regret property --- they do not know the precise algorithm/heuristic used). 

Our main result designs an auction that extracts revenue equal to the \emph{full expected welfare} whenever the buyers are ``mean-based'' (a property satisfied by standard no-regret learning algorithms such as Multiplicative Weights, Follow-the-Perturbed-Leader, etc.). This extends a main result of~\cite{BravermanMSW18} which held only for a single buyer. 

Our other results consider the case when buyers are mean-based but never overbid. On this front, \cite{BravermanMSW18} provides a simple LP formulation for the revenue-maximizing auction for a single-buyer. We identify several formal barriers to extending this approach to multiple buyers.
\end{abstract}
\addtocounter{page}{-1}
\newpage
\section{Introduction}\label{sec:intro}
Classical Bayesian auction design considers a static auction where buyers participate once. Here, the study of truthful auctions is ubiquitous following Myerson's seminal work~\cite{Myerson81}. But many modern auction applications (such as ad auctions) are \emph{repeated}: the same buyers participate in many auctions over time. Moreover, the vast majority of auction formats used in such settings are \emph{not} truthful (e.g.~first-price auctions, generalized first-price auctions, generalized second-price auctions). Even those that are based on a truthful format (such as the Vickrey-Clarke-Groves mechanism~\cite{Vickrey61,Clarke71,Groves73}) are no longer truthful when the repeated aspect is taken into account (because the seller may increase or decrease reserves in later rounds based on buyers' behavior in earlier rounds). As such, it is imperative to have a study of non-truthful repeated auctions.

Over the past several years, this direction has seen significant progress on numerous fronts (we overview related work in Section~\ref{sec:related}). Our paper follows a recent direction initiated by~\cite{BravermanMSW18} and motivated by empirical work of~\cite{NekipelovST15}. Specifically,~\cite{NekipelovST15} find that bidding behavior on Bing largely satisfies the no-regret guarantee (that is, there exist values for the buyers such that their bidding behavior guarantees low regret --- the paper makes no claims about any particular algorithm the buyers might be using). This motivates the following question: \emph{if buyer behavior guarantees no-regret, what auction format for the designer maximizes her expected revenue}?

\cite{BravermanMSW18} initiated this study for a single buyer. The main focus of our paper is to initiate the study for multiple buyers. We formally pose the model in Section~\ref{sec:prelim}, and overview our main results here. 

The concept of a ``mean-based'' no-regret learning algorithm appears in~\cite{BravermanMSW18}, and captures algorithms which pull an arm with very high probability when it is historically better than all other arms (formal definition in Section~\ref{sec:prelim}). While it is common to design non-mean-based algorithms for dynamic environments, standard no-regret algorithms such as Multiplicative Weights, EXP3, etc. are all mean-based.\footnote{Note also that these canonical algorithms are mean-based even if the learning rate changes over time, as long as the learning rate is $\omega(1/T)$.}\\

\noindent\textbf{Main Result} (informal --- See Theorem~\ref{thm:overBiddingTheorem}:) When any number of i.i.d.~buyers use bidding strategies which satisfy the mean-based no-regret guarantee, there exists a repeated single-item auction for the seller which guarantees them expected revenue arbitrarily close to the optimal expected \emph{welfare}.\\

One main result of~\cite{BravermanMSW18} proves the special case with just a single buyer. While we defer technical details of our construction to Section~\ref{sec:FSE}, we briefly overview the main challenges here. The one-buyer~\cite{BravermanMSW18} auction is already surprising, as it requires the seller to both (a) give the buyer the item every round, yet (b) charge them their full value (without knowing their value). The key additional challenge for the multi-buyer setting is that the seller must now give the item not just to \emph{a} buyer in every round, but to \emph{the} buyer with the highest value. This means, in particular, that we must set up the auction so that buyers will pull a distinct arm for each possible value, and yet we must still charge each buyer their full expected value by the end of the auction.

Our auction, like that of~\cite{BravermanMSW18}, is fairly impractical (for example, it alternates between running a second-price auction every round, and charging huge surcharges to the winner) and is not meant to guide practice. Still, Theorem~\ref{thm:overBiddingTheorem} establishes that full surplus is possible for multiple mean-based buyers, and therefore sets a high benchmark for this setting without further modeling assumptions.\\

\noindent\textbf{No Overbidding.} Indeed, the main impracticality in our Full-Surplus-Extraction auction is that lures buyers into \emph{overbidding} significantly, and eventually paying more than their value. In practice, it may be reasonable for buyers to be \emph{clever}, and just remove from consideration all bids exceeding their value (but guarantee no-regret on the remaining ones). To motivate this, observe that overbidding is a dominated strategy in all the aforementioned non-truthful auctions. So we next turn to analyze auctions for clever, mean-based buyers. 

Here, the second main result of~\cite{BravermanMSW18} characterizes the revenue-optimal repeated auction via a linear program, and shows that it takes a particularly simple form. On this front, we identify \emph{several} formal barriers to extending this result to multiple buyers. Specifically:
\begin{itemize}
\item For a single buyer,~\cite{BravermanMSW18} write a concise polytope (we call it the `BMSW polytope') characterizing auctions which can be implemented for a single clever mean-based buyer (i.e.~being in this polytope is necessary and sufficient to be implementable). We show that two natural extensions of this polytope to multiple buyers contain auctions which \emph{cannot} be implemented for multiple clever mean-based buyers (we show that being in either natural polytope is necessary, but not sufficient). This is in Section~\ref{sec:LP}.
\item For a single buyer,~\cite{BravermanMSW18} shows that the optimal auction is ``pay-your-bid with declining reserve.\footnote{That is, each round there is a reserve. Any bid above the reserve wins the item, but pays their bid. The reserve declines over time.}'' We show that a natural generalization of this ``pay-your-bid uniform auctions with declining reserve'' to multiple buyers captures many extreme points of the multi-buyer BMSW polytope. But, we also show that such auctions are not necessarily optimal (meaning that this aspect of~\cite{BravermanMSW18} does not generalize to multiple buyers either). This is in Section~\ref{sec:UPYB}.
\item Finally, we establish that not only does the particular multi-buyer BMSW polytope not capture all implementable auctions for clever mean-based buyers, but the space of implementable auctions is \emph{not even convex}! This is in Section~\ref{sec:nonconvex}.
\end{itemize}

While our results are not a death sentence for the future work in the~\cite{BravermanMSW18} model for clever mean-based buyers, the barriers do shut down most natural multi-buyer extensions of their approach. Still, these barriers also help focus future work towards other potentially fruitful approaches, which we highlight in Section~\ref{sec:conclusion}.

\subsection{Related Work}\label{sec:related}
There is a vast body of work at the intersection of learning and auction design. Much of this considers learning from the perspective of the \emph{seller} (e.g.~sample complexity of revenue-optimal auctions), and is not particularly related~\cite{DhangwatnotaiRY15,ColeR14,HuangMR15,DevanurHP16,RoughgardenS16,MorgensternR15,MorgensternR16,GonczarowskiN17,CaiD17,HartlineT19,GonczarowskiW18,GuoHZ19,GuoHTZ19,BrustleCD20}.

More related is the recent and growing literature on dynamic auctions~\cite{PapadimitriouPPR16,AshlagiDH16,MirrokniLTZ16,LiuP18,MirrokniLTZ18,MirrokniLTZ19} Like our model, the auction is repeated. The distinction between these works and ours is that they assume the buyer is \emph{fully strategic} and processes fully how their actions today affect the seller's decisions tomorrow (whereas we instead model buyers as no-regret learners). 

The most related work to ours is in the~\cite{BravermanMSW18} model itself. 
Here we provide a brief summary of the main results in~\cite{BravermanMSW18} and their connection to our main results. \cite{BravermanMSW18} studies the one seller one buyer scenario, where the buyer employs a mean-based no-regret algorithm. The authors present three results, each obtained under different assumptions regarding the behavior of the buyers.
Firstly (as we have already mentioned earlier in the introduction), \cite{BravermanMSW18} shows that for vanilla mean-based no-regret buyers, \cite{BravermanMSW18} can extract revenue that is an arbitrarily large fraction of the bidder's expected value. Our Theorem~\ref{thm:overBiddingTheorem} extends this result to the multiple buyer setting, overcoming novel technical and conceptual challenges. Second, \cite{BravermanMSW18} designs a novel (not mean-based) learning algorithm against which the optimal mechanism for the seller is simply Myerson's auction in each round. Their proof of this result naturally accommodates  multiple buyers. Finally, \cite{BravermanMSW18} shows that if the buyer is clever and mean-based no regret (where they do not overbid their value), then the optimal auction has a clean tractable format (pay-your-bid with declining reserve over time). As we have discussed in the ``No Overbidding" section of the introduction, our work shows several formal barriers in extending these results to  multiple buyers. In summary, our main result extends their first main result to multiple bidders. Their second result already holds for multiple bidders (so there is nothing for us to extend). Our secondary results establish formal barriers to extending their final main result to multiple bidders. 

Two recent follow-ups have extended 
the setting in \cite{BravermanMSW18} in a different direction. First, \cite{DengSS19} considers the problem of playing a two-player game against a no-regret learner. While technically not an auctions problem, there is thematic overlap with our main result.~\cite{DengSS19a} extends the single-buyer results in~\cite{BravermanMSW18} to be {\emph{prior-free}. Specifically, they show how to design auctions achieving the same guarantees as those in~\cite{BravermanMSW18} but where the buyer's values are chosen adversarially.} In comparison to these works, ours is the first to extend the model to consider multiple buyers.

Finally, recent work of~\cite{CamaraHJ20} considers interaction between a learning buyer and a \emph{learning} seller. Their seller does not have a prior against which to optimize, and instead itself targets a no-regret guarantee. In comparison, our seller (like the seller in all previously cited works) optimizes expected revenue with respect to a prior.

\section{Preliminaries}\label{sec:prelim}
We consider the same setting as~\cite{BravermanMSW18}, extended to multiple buyers. Specifically, there are $n$ buyers and $T$ rounds. In each round, there is a single item for sale. Each buyer $i$ has value $v_{i,t}$ for the item during round $t$, and each $v_{i,t}$ is drawn from $\dist$ independently (that is, the buyers are i.i.d., and the rounds are i.i.d.~as well). For simplicity of exposition (and to match prior work), we assume $\dist$ has finite support $0 \leq w_1 < w_2 < \ldots < w_m \leq 1$ and we define $q_j$ to be the probability $w_j$ is drawn from $\dist$.

Each round, the seller presents $K$ arms for the buyers. Each arm is labeled with a bid, and we assume that one of the arms is labeled with $0$ (to represent a bid of ``don't participate''). Note that the same set of arms is presented to all buyers, and the same set of arms is presented in each round. 

In each round $t$, the seller defines an anonymous auction. Specifically, for all $i,t$, the seller defines $a_{i,t}(\vec{b})$ to be the probability that buyer $i$ gets the item in round $t$, and $p_{i,t}(\vec{b}) \in [0,b_i\cdot a_{i,t}(\vec{b})]$ to be the price buyer $i$ pays, when each buyer $j$ pulls the arm labeled $b_j$. To be anonymous, it must further be that for all permutations $\sigma$ of the buyers that $(a_{\sigma(i),t}(\sigma(\vec{b})),p_{\sigma(i),t}(\sigma(\vec{b})) = (a_{i,t}(\vec{b}),p_{i,t}(\vec{b}))$ (the auction is invariant under relabeling buyers). The only additional constraints on $a$ are that $\sum_i a_{i,t}(\vec{b}) \leq 1$, for all $t,\vec{b}$ (item can be awarded at most once), and that $b'_i > b_i \Rightarrow a_{i,t}(b_i;\vec{b}_{-i}) \leq a_{i,t}(b'_i;\vec{b}_{-i})$ for all $i,\vec{b}_{-i},b_i,b'_i$ (allocation is monotone). $p$ must also be monotone ($b'_i > b_i \Rightarrow p_{i,t}(b_i;\vec{b}_{-i}) \leq p_{i,t}(b'_i;\vec{b}_{-i})$). When we state prior work in the single-buyer setting, we may drop the buyer subscript of $i$ (for instance, we will write $a_{1, t} (b_1)$ as $a_{t}(b)$).
\subsection{Contextual Bandits}
Like~\cite{BravermanMSW18}, we model the buyers as online learners. Also like~\cite{BravermanMSW18}, our results apply equally well to the experts and bandits model, where $v_{i,t}$ serves as buyer $i$'s context for round $t$. Specifically:

\begin{itemize}
\item For all subsequent definitions below, fix a buyer $i$, fix a bid vector $\vec{b}_{-i,t}$ for all rounds $t$, and fix $a_{i,t}(\cdot)$. 
\item For any bid $b$, buyer $i$, and round $t$, define $r_{ibt}(v):=v\cdot a_{i,t}(b;\vec{b}_{-i}) - p_{i,t}(b;\vec{b}_{-i})$. That is, define $r_{ibt}(v)$ to be the utility during round $t$ that buyer $i$ would enjoy by bidding $b$ with value $v$.
\item For an algorithm $S$ (decides a bid for round $t$ based only on what it observes through rounds $t-1$, and its value $v_{i,t}$)\footnote{In the bandits model, buyer $i$ learns only $r_{ibt}(v)$ for the bid $b:=b_{it}$ after each round $t$ (and all $v$). In the experts model, buyer $i$ learns $r_{ibt}(v)$ for all $b$ (and all $v$).} that submits bids $b_{it}$ in round $t$, its total payoff is $P(S):=\mathbb{E}[\sum_t r_{ib_{it}t}(v_{i,t})]$. The expectation is over any randomness in the bids $b_{it}$, as $S$ may be a randomized algorithm, and the randomness in $v_{i,t}$.
\item An algorithm is fixed-bid if $v_{it} = v_{it'} \Rightarrow b_{it} = b_{it'}$. That is, the algorithm may make different bids in different rounds, but only due to changes in the buyer's value. Let $\mathcal{F}$ denote the set of all fixed-bid strategies.
\item The \emph{regret} of an online learning algorithm $S$ is $\max_{F \in \mathcal{F}}\{P(F) - P(S)\}$.
\item An algorithm is $\delta$-low regret if it guarantees regret at most $\delta$ \emph{on every fixed sequence of auctions, and fixed bids of the other players}. We say that an algorithm is \emph{no-regret} if it is $\delta$-low regret for some $\delta = o(T)$.
\end{itemize} 

Like~\cite{BravermanMSW18}, we are particularly interested in algorithms ``like Multiplicative Weights Update:''

\begin{definition}[Mean-Based Online Learning Algorithm,~\cite{BravermanMSW18}] Let $\sigma_{i,b,s}(v):=\sum_{t < s}r_{ibt}(v)$. An algorithm is \emph{$\gamma$-mean-based} if whenever $\sigma_{i,b,s}(v_{i,s}) < \sigma_{i,b',s}(v_{i,s}) - \gamma T$ (for any $b,b'$), then the probability that the algorithm bids $b$ during round $s$ is at most $\gamma$. An algorithm is \emph{mean-based} if it is $\gamma$-mean-based for some $\gamma = o(1)$.
\end{definition}

As noted in~\cite{BravermanMSW18}, natural extensions of Multiplicative Weights, EXP3, Follow the Perturbed Leader, etc. to the contextual setting are all mean-based online learning algorithms. 

\subsection{Learners and Benchmarks}
Before formally stating our main results, we first provide relevant benchmarks. We use $\val_n(\mathcal{D}):=\mathbb{E}_{\vec{v}\leftarrow \mathcal{D}^n}[\max_i v_i]$ to denote the expected maximum value among the $n$ buyers. We use $\rev_n(\mathcal{D})$ to denote the expected revenue of the optimal truthful auction when $n$ buyers have values drawn from $\mathcal{D}$. We make the following quick observation, which holds for \emph{any} low regret learning algorithm (and extends an observation made in~\cite{BravermanMSW18} for a single buyer).

\begin{observation} \label{obs:noRegretRevenueUpperBound} The seller cannot achieve expected revenue beyond $T \cdot \val_n(\mathcal{D}) + o(T)$ when buyers guarantee no-regret, even if the seller knows precisely what algorithms the buyers will use.
\end{observation}
\begin{proof}
Any strategy for the seller must allow the buyers to bid $0$ in every round and not get the item. This means that the best fixed-bid in hindsight for every buyer guarantees utility at least $0$. Therefore, if the buyer is no-regret, they must have utility at least $-o(T)$. The total expected welfare the buyers can have together is at most $T \cdot \val_n(\mathcal{D})$. As revenue = welfare $-$ utility, this means that revenue can be at most $T \cdot \val_n(\mathcal{D}) + o(T)$.
\end{proof}

Finally, we will consider two types of no-regret learners. One type we will consider is simply no-regret learners who use a mean-based learning algorithm. Second, we will consider no-regret learners who use a no-regret learning algorithm but \emph{never overbid}. Specifically, such learners immediately remove from consideration all bids $b_{it} > v_{i,t}$, but otherwise satisfy the no-regret guarantee. We refer to such learners are \emph{clever}.~\cite{BravermanMSW18} motivate such learners by observing that in most (perhaps all) standard non-truthful auction formats, overbidding is a \emph{dominated strategy}. For example, it is always better to bid truthfully than to overbid in a first-price auction, generalized first-price auction, generalized second-price auction, and all-pay auction. 

\subsection{Border's Theorem}
Some of our work will use Border's theorem~\cite{Border91}, which considers the following. Consider a monotone, anonymous (not necessarily truthful) single-item auction, and a fixed strategy $s(\cdot)$ which maps values to actions. Let $x(w_j)$ denote the probability that a buyer using action $s(w_j)$ wins the item, assuming that all other buyers' values are drawn i.i.d.~from $\dist$ and use strategy $s$ as well. Border's theorem asks the following: when given some vector $\langle x_1,\ldots, x_m\rangle$, does there exists a monotone anonymous (not necessarily truthful) single-item auction such that $x(w_j) = x_j$ for all $j$? If so, we say that $\vec{x}$ is \emph{Border-feasible}. Below is Border's theorem. We will not actually use the precise Border conditions in any of our proofs, just the fact that they exist and are linear in $\vec{x}$.

\begin{theorem}[Border's Theorem~\cite{Border07,cai2011constructive,CheKM11}]\label{thm:Border} {When the buyers are drawn i.i.d from $\D$ (meaning each buyer's probability of valuing the item at $w_j$ is $q_j$), } $\vec{x}$ is Border-feasible if and only if it satisfies the \emph{Border conditions}:
\begin{align*}
 n\sum_{\ell \geq j}q_j \cdot x_j \leq 1 - (1-\sum_{\ell \geq j} q_j)^n \enspace \forall j \in [m].
\end{align*}

\end{theorem}

\section{Full Surplus Extraction from Mean-Based buyers}\label{sec:FSE}
Here, we show a repeated auction which achieves expected revenue arbitrarily close to $T \cdot \val_n(\dist)$ when buyers are mean-based (but consider overbidding). We also note that our auction does not depend on the particular mean-based algorithms used. The auction does \emph{barely} depend on $\dist$, but only in initial ``setup rounds'' (the auction during almost all rounds does not depend on $\dist$). Recall this guarantee is the best possible, due to Observation~\ref{obs:noRegretRevenueUpperBound}. 

\begin{theorem}\label{thm:overBiddingTheorem}
	Whenever $n$ buyers use strategies satisfying the mean-based guarantee, there exists a repeated auction which obtains revenue $T \cdot (1 - \delta)\val_n(\D) - o(T)$ for any constant $\delta < 1$.
\end{theorem}

In this language, one main result of~\cite{BravermanMSW18} proves Theorem~\ref{thm:overBiddingTheorem} when $n=1$. Before diving into our proof, we remind the reader of the main challenge. In order to possibly extract this much revenue, the auction must somehow both (a) charge each winning buyer their full value, leaving them with zero utility, yet also (b) figure out which buyer has the highest value in each round, and give them the item. The distinction between the $n=1$ and $n>1$ case is in (b). When $n=1$, it is still challenging to give the buyer the item every round while charging their full value, but at least the buyer does not need to convey any information to the seller (so, for example, it is not necessary to incentivize the buyer to pull distinct arms for each possible value --- the buyers just need to pay their full value on average by the end). When $n>1$, we need the buyer to pull a distinct arm for each of their possible values, because we need to make sure that the highest buyer wins the item (and the only information we learn about each buyer's value is the arm they pull).\\

\noindent\textbf{Additional Notation.} We now provide our auction and analysis, beginning with some additional notation for this section. We will divide the $T$ rounds of the auction into {\em phases} of $2R$ consecutive rounds, where $R = \Omega(T)$. There are $P:=T/(2R)$ total phases (so $P$ is a constant, but it will be a large constant depending on $\delta$). In our contruction, the first $m-1$ phases will be the setup phases and the last $P-m+1$ phases will be the main phases. The goal of the setup phases is to align buyer's incentives so that they will behave in a particular manner in later phases. The main phases are where we will extract most of our revenue.

Recall that there are $K$ non-zero arms labeled $b_1 < \ldots < b_K$. {Our construction will use $K:=P$ arms.} Because the buyers consider overbidding, the precise bid labels are not important, so long as they are sufficiently large (concretely, we set $b_i:=2w_m+i$). We will sometimes index arms using $b_j^\phaseVar:=b_{P+j-\phaseVar}$. This notation will be helpful to remind the reader that $b_j^\phaseVar$ is the arm that we intend to be pulled by a buyer with value $w_j$ during {main} phase $\phaseVar$.\footnote{So for example, arm $b_{P-m}$ will first be (intended to be) pulled by buyers with value $w_1$ in phase $m+1$, then by buyers with value $w_2$ in phase $m+2$, etc.} 

\subsection{Defining the Auction} \label{armdef} Intuitively, our auction tries to do the following. In each phase $\phaseVar$, there is a targeted arm $b_j^\phaseVar$ for each possible value $w_j$, so there are $m$ arms that are (intended to be) pulled during each phase. {Ideally, since $w_j$ needs to transition from pulling $b_j^{\phaseVar-1}$ in phase $\phaseVar - 1$ to pulling $b_j^{\phaseVar}$ in phase $\phaseVar$,  at the beginning of phase $\phaseVar$, $w_j$ should be indifferent between pulling $b_j^{\phaseVar - 1}$, $w_j$'s favourite arm in phase $\phaseVar - 1$; and $b_j^{\phaseVar}$,$w_j$'s intended arm for phase $\phaseVar$. Let us for now assume this is true and see how we design the auction during phase $\phaseVar$ (which contains $2R$ rounds). } 

The base auction each round is just a second-price auction, where pulling arm $b_j^\phaseVar$ submits a bid of $w_j$. For the first $R$ rounds of each phase, this is exactly the auction executed. Because the second-price auction is dominant strategy truthful, this lures a mean-based buyer with value $w_j$ into having high cumulative reward for arm $b_j^\phaseVar$ (and in particular, strictly higher than any other arm). For the second $R$ rounds of each phase, the base auction will still be the same second-price auction, except we will now overcharge each buyer so that their average utility during all $2R$ rounds of auction in phase $\phaseVar$ is close to zero.
In principle, this is possible because the buyers have high cumulative utility for this arm from the first $R$ rounds, and are purely mean-based (and so they will pay more than their value to pull an arm which is historically much better than all others). 

{Now, by design our auction in phase $\phaseVar$ gives the item to the highest buyer most of the time, therefore the expected welfare is almost optimal. Meanwhile, the expected utility is close to 0, which means we have managed to extract revenue that is almost the full welfare in phase $\phaseVar$. Lastly, notice that cumulative utility for arm $b_j^{\phaseVar+1}$ increases during phase $\phaseVar$, so our phase cannot last forever. If we set the phase length to be too long, then $b_j^{\phaseVar + 1}$ will become $w_j$'s favourite arm before phase $\phaseVar$ ends. This is exactly why we need multiple phases instead of one phase. Let us set our phase length in such a way that at the end of $\phaseVar$, the increase in cumulative utility for arm $b_j^{\phaseVar+1}$ is just enough for $w_j$ to be indifferent between $b_j^\phaseVar$ and $b_j^{\phaseVar + 1}$, then the exact condition we assume at the start of phase $\phaseVar$ is satisfied, but for phase $\phaseVar + 1$. Thus we can safely start a new phase $\phaseVar + 1$ and extract almost full welfare by the same auction design. }

Of course, this is just intuition for why an auction like this could possibly work --- significant details remain to prove that it does in fact work (including precisely the choice of overcharges, analyzing incentives between phases, etc.). Below is a formal description of our auction.\footnote{There is one technicality, which is that we would like to always have a mean-based learner break ties to pull arm $b_0$ with probability $o(1)$ during any round where another arm guarantees them non-negative cumulative utility. To do this, we just decrease all payments defined below by an arbitrarily small $\varepsilon$, but omit this from the definition for cleanliness.}

\begin{definition}[Full Surplus Extraction Auction] The FSE Auction uses the following allocation and payment rule in each round. There are two steps in each round. First, based on the arm pulled, a bid is submitted on behalf of the buyer into a secondary auction. Then, the secondary auction is resolved. There are three types of arms:
\begin{itemize}
\item Some arms are \emph{dormant}. These arms don't enter the secondary auction (i.e.~no item and $0$ payment). 
\item Some arms are \emph{active}. Pulling arm $b_{P-\phaseVar+j} = b_{j}^\phaseVar$ enters a bid of $w_j$ into a secondary auction. 
\item Some arms are \emph{retired}. Pulling a retired arm enters a bid of $w_m+1$ into a secondary auction.
\end{itemize}

Which arms are dormant/active/retired change each phase. In addition, the secondary auction resolves differently for the first $m-1$ phases (we call these the setup phases) versus the last $P-m+1$ phases (we call these the main phases). Think of $P \gg m$, so the main phases are what matter most. In any main phase $(\phaseVar \geq m)$:
\begin{itemize}
\item Active arms: $b_{P-\phaseVar+1}=b_1^\phaseVar$ through $b_{P-\phaseVar+m}=b_m^\phaseVar$. Dormant: below $b_{P-\phaseVar+1}$. Retired: above $b_{P-\phaseVar+m}$. {Note that by our definition, the index of active arms decreases as $\phaseVar$ increases. For instance, if in $n^{th}$ phase the active arms are $b_{h}, b_{h-1}, \cdots, b_{l}$, then in the $n+1^{th}$ phase the active arms are $b_{h-1}, b_{h-2}, \cdots, b_{l-1}$.}
\item The secondary auction awards the item to a uniformly random buyer who submits the highest bid.
\item If the winning arm was retired (i.e.,~submitted a bid of $w_m+1$), they pay $2w_m$.
\item If the winning arm was active, the winner pays the second-highest bid.
\item Additionally, in the second $R$ rounds of a phase, if the highest bid is $w_j$ and the second-highest bid is $w_\ell$, the winner pays an additional surcharge of $2(w_j-w_\ell)$.
\end{itemize}
In any setup phase $(\phaseVar < m)$:
\begin{itemize}
\item Active arms: $b_{P-\phaseVar+1}=b_1^\phaseVar$ through $b_P=b_\phaseVar^\phaseVar$. Dormant: below $b_{P-\phaseVar+1}$. Retired: none.
\item Let $S$ denote the set of buyers who submitted the highest bid. The secondary auction picks $i^*$ uniformly at random from $S$ as a tentative winner.
\item If the highest bid is $< w_\phaseVar$, then $i^*$ wins the item.
\item If the highest bid is $w_\phaseVar$, then $i^*$ wins the item with probability\footnote{Recall that $q_j$ is the probability a buyer's value for the item is $w_j$.} $(q_\phaseVar/\sum_{j \geq \phaseVar} q_j)^{|S|-1}${, which is equal to the probability that $i^*$ wins the item had the buyers participated in a second price auction and $i^*$ bids $w_\phaseVar$ while other buyers bid truthfully}. 
\item If there is a winner, the winner pays the second-highest bid. 
\item Additionally, in the second $R$ rounds of a phase, if the highest bid is $w_j$ and the second-highest bid is $w_\ell$, the winner (if there is one) pays an additional surcharge of $2(w_j-w_\ell)$.
\end{itemize}
\end{definition}

The main phases are the interesting part of the FSE auction where we will argue that mean-based buyers choose to pull their intended arm and get $0$ utility each phase for doing so. The setup phases are just a technical setup to get incentives to work out, and is the reason why we lose an arbitrarily small $\delta$ fraction of the optimal welfare (because during the setup phases, we aren't giving the item to the highest buyer --- but fortunately there are just $m$ of $P$ setup phases, and we can take $P$ to be any large constant $\gg m$ that we want). {A concrete execution of the FSE auction is given in Appendix~\ref{app:example}.}

We first quickly confirm that the FSE Auction is monotone (proof in Appendix~\ref{app:fse}).
\begin{observation}\label{obs:monotone} The allocation and payment rule for the FSE auction are both monotone.
\end{observation}

\subsection{Mean-Based Behavior}
Before analyzing the expected revenue of the seller, we first analyze the behavior of mean-based buyers. The challenge, of course, is that the payoff from each arm depends on the behavior of the other buyers, who are themselves mean-based. So our goal is to establish that mean-based learning in the FSE auction forms some sort of ``equilbrium,'' in the sense that one mean-based buyer pulls the desired arm almost-always provided that all other buyers pull the desired arm almost-always. Our first step is characterizing a buyer's payoff for each arm at each round, assuming that all other buyers pull the intended arm almost always. 

The main steps in our proof are as follows. First, we analyze the cumulative payoff for a buyer with each possible value for each possible arm, assuming that each other buyer pulls their intended arm. We then conclude that a buyer with value $w_j$ has highest cumulative utility for their intended arm for the entirety of each phase. However, we also establish that the utility they enjoy during each phase for their intended arm is $0$. This means that every buyer has $0$ utility at the end (up to $o(T)$), meaning that the seller's revenue is equal to the expected welfare. Because we give the item to the highest value buyer whenever they pull the intended arm, the welfare is $T \cdot \val(\dist)$. We now proceed with each step.

In each of the technical lemmas below, we let $H_s(v,b)$ denote the cumulative payoff {during rounds 0 to s} that a buyer with value $v$ would have enjoyed in hindsight by pulling arm $b$ in the FSE Auction, assuming that all other buyers pull their intended arm for at least a $1-o(1)$ fraction of the rounds during every main phase $\phaseVar$, and that they pull either their intended arm (if it exists) or $b_P$ (otherwise) during every setup phase $\phaseVar$. {We let  $\XVCG(v)$ denote the probability that a bidder with value $v$ wins a second-price auction when bidding truthfully against $n-1$ values drawn independently from $\dist$ (ties broken randomly). And we let $\PVCG(v)$ denote the interim payment made by a bidder with value $v$ to a second-price auction, in expectation over $n-1$ other values drawn independently from $\dist$.\footnote{Formally, let $X_i$ be independent draws from $\dist$ for $i = 1$ to $n-1$. Define $X_0:=v$. Let $X:=\max_{i \geq 1}\{X_i\}$, and $Y$ be an indicator random variable for the event that a uniformly random element in $\arg\max_{i \geq 0}\{X_i\}$ is $0$. Then $\PVCG(v):=\mathbb{E}[X\cdot Y]$.} } 

Our first lemma will concern the cumulative payoff for arms at the start of each phase.

\begin{lemma}\label{lem:firstFSE} At the end of phase $\phaseVar$, the change in cumulative payoff of a buyer with value $v$ for each arm satisfies:
\begin{itemize}
\item For dormant arms $b$, $H_{2R\phaseVar}(v,b) - H_{2R(\phaseVar-1)}(v,b)= 0$.
\item For active arms: $H_{2R\phaseVar}(v,b_j^\phaseVar) - H_{2R(\phaseVar-1)}(v,b_j^\phaseVar)= 2R \cdot (v-w_j) \cdot \XVCG(w_j) \pm o(T)$.
\item For retired arms: $H_{2R\phaseVar}(v,b_j)  - H_{2R(\phaseVar-1)}(v,b_j)= 2R\cdot (v - 2w_m) \pm o(T)$.
\end{itemize} 
\end{lemma}
\begin{proof}[Proof Sketch]
We provide a complete proof here in the case that other buyers pull the stated arm for \emph{all} rounds (rather than a $1-o(1)$ fraction), and only for the main phases (rather than also for setup phases). This captures the intuition, but postpones the technical challenges. In Appendix~\ref{app:fse}, we provide a complete proof for the general case. 

The first bullet follows trivially. To see this, observe that if an arm is dormant for phase $\phaseVar$, then it never awards the item or charges payments, so the cumulative utility doesn't change.

For active arms, observe that pulling arm $b_j^\phaseVar$ enters bid $w_j$ into a second price auction where all other bids are drawn independently from $\dist$. This means that the buyer wins the item with probability $\XVCG(w_j)$ in each such round. 

In terms of payment, observe that the expected payment during the first $R$ rounds is exactly $\PVCG(w_j)$ per round. Moreover, the expected payment during the second $R$ rounds is exactly $\XVCG(w_j) \cdot (2w_j - w_\ell) = 2w_j\cdot \XVCG(w_j) - \PVCG(w_j)$ per round. Therefore, the added utility during phase $\phaseVar$ is:
\begin{align*}
 &2R\cdot v \cdot \XVCG(w_j) - R \cdot \PVCG(w_j) - R\cdot (2w_j \cdot \XVCG(w_j)) + R\cdot \PVCG(w_j) \\
 = &2R \cdot (v-w_j) \cdot \XVCG(w_j) .
\end{align*}

For retired arms, observe that pulling the arm enters bid $w_m+1$ into a second-price auction, which surely wins in all rounds. The payment is $2w_m$. Therefore, the change in cumulative payoff is $2Rv - 4Rw_m$.
\end{proof}

The above proof captures the main intuition for Lemma~\ref{lem:firstFSE}, but non-trivial extra work is required for the two mentioned extensions. Still, full proof follows the above outline with one extra step. While it is generally \emph{extremely} unwieldy to predict the behavior of multiple no-regret buyers playing a repeated game (even when those buyers are guaranteed to play a particularly simple no-regret strategy like Multiplicative Weights Update), our repeated game is especially simple because buyers have a weakly dominant strategy (for the rounds which are just a second-price auction). Therefore, we can argue that any mean-based strategies will quickly learn to play their intended arm, so long as it starts the phase within $o(T)$ of their favorite. This intuition will be made formal when we prove Lemma~\ref{lem:setupfirst} (proof contained in Appendix~\ref{app:fse}). 

Using Lemma~\ref{lem:firstFSE}, we can also quickly conclude the cumulative utilities. To help parse the formulas below, recall that every arm transitions from being initially dormant, to being intended for $w_1$, then $w_2$, etc., and then retired. However, some arms start out in the middle of this transition (e.g. arm $b_{P+m-1}$ is initially intended for $w_m$, and then retires), and some arms end in the middle of this transition (e.g. arm $b_1$ is dormant until the final phase, when it is intended for $w_1$). For the rest of the formally stated corollaries/lemmas in this section, we will use terms like $o(T)$ to keep the statements clean. However, the precise choice of $o(T)$ that we prove will be relevant for how these claims are used in later proofs, and versions of these statements in Appendix~\ref{app:fse} are more precise. We include them here to help give a clear outline for the rest of the proof.

\begin{corollary}\label{cor:firstFSE} At the end of phase $\phaseVar$, the cumulative payoffs for a buyer with value $v$ satisfy:
\begin{itemize}
\item If $b_j$ is dormant during phase $\phaseVar$ ($j \leq P-\phaseVar$): $H_{2R\phaseVar}(v,b)= 0$.
\item If $b_j$ is active during phase $\phaseVar$ ($j \in [P-\phaseVar +1,P-\phaseVar+m]$):
$$H_{2R\phaseVar}(v,b_j) = 2R \cdot \left(\sum_{k=1}^{j + \phaseVar - P}(v-w_k) \cdot \XVCG(w_k) \right)\pm o(T).$$
\item If $b_j$ is retired during $\phaseVar$ ($j \geq P-\phaseVar + m+1$):
$$H_{2R\phaseVar}(v,b) = 2R \cdot \left((\phaseVar-m)\cdot (v - 2w_m) + \sum_{k=1}^m  (v-w_k) \cdot \XVCG(w_k)  \right)\pm o(T).$$
\end{itemize}

\end{corollary}

Lemma~\ref{lem:firstFSE} and Corollary~\ref{cor:firstFSE} now starts to make our analysis transparent: Consider any particular value $w_j$ and their cumulative utility for any particular arm $b_\ell$. During every phase that $b_\ell$ is dormant, the cumulative utility doesn't change. During every phase that $b_\ell$ is retired, the cumulative utility strictly decreases. During every phase that $b_\ell$ is intended for a value $w_k$, with $k > j$, the cumulative utility strictly decreases (because $\XVCG(w_k) > 0$ for all $k$). During every phase that $b_\ell$ is intended for a value $w_k$, with $k < j$, the cumulative utility strictly increases. During the phase where $b_\ell$ is intended for $w_j$, the cumulative utility stays the same. We can now establish that mean-based buyers pull the intended arm during each phase.

\begin{lemma}\label{lem:favestart} For all $\phaseVar$, at the start of each phase $\phaseVar$, when $j \leq \phaseVar$, a buyer with value $w_j$ has highest cumulative utility for arm $b_j^\phaseVar$, and also $b_j^{\phaseVar-1}$. Specifically, for all other arms $b_\ell$: 
$$H_{2R(\phaseVar-1)}(w_j,b_j^{\phaseVar-1}) \pm o(T) = H_{2R(\phaseVar-1)}(w_j,b_j^\phaseVar) > H_{2R(\phaseVar-1)}(w_j,b_\ell) + \Omega(T).$$ When $j > \phaseVar$, for all $b_\ell \neq \phaseVar$,  $H_{2R(\phaseVar-1)}(w_j,b_\phaseVar^\phaseVar) > H_{2R(\phaseVar-1)}(w_j,b_\ell) + \Omega(T).$
\end{lemma}

\begin{lemma}\label{lem:setupfirst} For all $\phaseVar$, assuming that all other buyers pull their intended arm except for $o(T)$ rounds, a mean-based buyer with value $w_j$ pulls arm $b_j^\phaseVar$ (if it exists) for the first $R$ rounds, except for at most $o(T)$ rounds. Otherwise, they pull arm $b_P = b_{\phaseVar}^{\phaseVar}$ for the first $R$ rounds, except for at most $o(T)$ rounds.
\end{lemma}

\begin{lemma}\label{lem:setupend} For all $\phaseVar$, assuming that all other buyers pull their intended arm except for $o(T)$ rounds, a mean-based buyer with value $w_j$ pulls arm $b_j^\phaseVar$ (if it exists) for the last $R$ rounds, except for at most $o(T)$ rounds. Otherwise, they pull arm $b_P = b_{\phaseVar}^{\phaseVar}$ for the last $R$ rounds, except for at most $o(T)$ rounds.
\end{lemma}

Finally, we combine everything together to conclude the following:

\begin{proposition}\label{prop:behave} When all buyers are mean-based, they all pull their intended arm in the FSE Auction, except for at most $o(T)$ rounds.
\end{proposition}

\subsection{Analyzing the Revenue}
Finally, we show that when all buyers pull their intended arm, the FSE auction extracts full surplus.

\begin{proof}[Proof of Theorem~\ref{thm:overBiddingTheorem}]
Except for the setup phases, and for rounds where buyers do not pull their intended arm, the auction gives the item to the highest buyer. Therefore, the expected welfare of the auction is at least $(1-m/P)T \cdot \val(\dist) - o(T)$. Moreover, Lemma~\ref{lem:firstFSE} establishes that through an entire phase, the cumulative utility of a buyer for pulling their intended arm is $0 \pm o(T)$. Therefore, the total utility of the mean-based buyer is at most $o(T)$. Therefore, the revenue is at least $(1-m/P)T \cdot \val(\dist) - o(T)$. Setting $P\geq m/\delta$ completes the proof.
\end{proof}

\section{Clever Mean-Based Buyers}\label{sec:UB}
In this section we consider clever mean-based buyers, and identify three formal barriers to developing optimal auctions for multiple clever mean-based buyers. We develop each barrier in the subsections below. Section~\ref{sec:LP} reminds the reader of the~\cite{BravermanMSW18} Linear Program, which exactly captures the optimal auction for a single clever mean-based buyer, and provides a natural extension to multiple buyers. 

\subsection{A Linear Programming Upper Bound}\label{sec:LP}

We first remind the reader of the~\cite{BravermanMSW18} Linear Program, and give a natural extension to multiple buyers. We first explicitly define variables for the results of a repeated auction.

Let $A$ be a repeated auction with $n$ i.i.d buyers of value distribution $\dist$. For each buyer $i$, let $S_i$ denote a strategy which takes as input a value $v_{it}$ for round $t$ (and all other information available from previous rounds) and outputs an arm $b^{S_i}_{it}(v_{it})$ to pull in round $t$. Let $\vec{v} := \langle v_{it}\rangle_{i \in [n], t \in T}$, which is drawn from the product distribution $\times_{nT}\dist$. We use the following notation:

\begin{align*}
& Rev_{A}(\dist, S_1,\ldots, S_n) := \E_{\vec{v}}\left[\sum_{i=1}^n \sum_t p_{i,t}\Big(b^{S_i}_{it}(v_{it});b^{S_{-i}}_{-it}(v_{-it})\Big)\right] \text{\quad\quad\quad \spc(expected revenue of the auction);}\\
&Rev_{n}(\dist, S_1, \ldots, S_n) := \max_{A}\{ Rev_{A}(\dist, S_1, \ldots, S_n) \} \text{\quad\quad\quad\quad\quad\spc (maximum attainable revenue for $S_1...S_n$)};\\
&X_{ij}^A(\dist, S_1, \ldots, S_n) = \frac{1}{T} \E_{\vec{v}}\left[\sum_t a_{it}\Big(b^{S_i}_{it}(w_j);b^{S_{-i}}_{-it}(v_{-it})\Big)\right]  \text{\spc\spc (average allocation probability when $v_i = w_j$)};\\
&Y_{ij}^A(\dist, S_1, \ldots , S_n) = \frac{1}{T} \E_{\vec{v}}\left[\sum_t a_{it}\Big(w_j; b^{S_{-i}}_{-it}(v_{-it})\Big)\right] \text{\quad\quad\spc (average allocation probability when bid is $w_j$)};\\
&U_{ij}^A(\dist, S_1,\ldots, S_n) = \frac{1}{T} \E_{\vec{v}}\left[\sum_t w_j\cdot a_{it}\Big(b^{S_i}_{it}(w_j); b^{S_{-i}}_{-it}(v_{-it})\Big) - p_{it}\Big(b^{S_i}_{it}(w_j); b^{S_{-i}}_{-it}(v_{-it})\Big)\right]\\
&\text{(average utility when $v_i = w_j$).}
\end{align*}

\begin{definition}[Auction Feasible]
A tuple of $m$-vectors $(\xst, \yst, \ust)$ is \textit{$n$-buyer auction feasible }
if there exists a repeated auction $A$, such that for all $\gamma = o(T)$, whenever $n$ buyers with values drawn i.i.d.~from $\dist$ run clever $\gamma$-mean-based strategies $S_1,\ldots, S_n$, then $\forall i, X_{ij}^A(\dist, S_1,\ldots, S_n)= \xst_j \pm O(\gamma); Y_{ij}^A(\dist, S_1,\ldots, S_n) = \yst_j \pm O(\gamma);U_{ij}^A(\dist, S_1,\ldots, S_n) =\ust_j  \pm O(\gamma)$. We call $(\xst, \ust)$ $n$-buyer auction feasible if there exists $\yst$ such that $(\xst, \yst, \ust)$ is $n$-buyer auction feasible. 
\end{definition}

One key insight in~\cite{BravermanMSW18} is that the space of auction feasible tuples is convex and can be characterized by simple linear equations. Below, note that the ``only if'' direction is slightly non-trivial, and we rederive it later for arbitrary $n$. The ``if'' direction requires designing an auction (for which we refer the interested reader to~\cite[Theorem~3.4]{BravermanMSW18}). We will not rederive the ``if'' direction, although we define the relevant auction later as well. 

\begin{theorem}[\cite{BravermanMSW18}]\label{thm:BMSWfeasible} $(x,u)$ is $1$-buyer auction feasible if and only if it satisfies the \emph{BMSW constraints}:\footnote{In fact, the `only if' portion of this theorem holds when replacing the clever mean-based buyer with just a clever buyer. But the `if' portion requires the stronger assumption of mean-based learning.}
\begin{align*}
    & u_i \geq (w_i - w_j) \cdot x_j, \enspace \forall i,j \in [m] : i > j,\\
    & x_i \geq x_j, \enspace \forall i\in [m], i > j,\\
    & u_i \geq 0, x_i \in [0,1], \enspace \forall i \in [m].
\end{align*}
\end{theorem}

Intuitively, the first BMSW constraint is the interesting one, which is necessary for the buyer to not regret pulling arm $b_j$ when their value is $w_i$ (again recall this is non-trivial, but we argue this shortly as a special case for general $n$). The second constraint is necessary because the auction must be monotone. The final constraint is necessary because the auction must have a null arm, and because all allocation probabilities must be in $[0,1]$ every round.

\cite{BravermanMSW18} also observe that the expected revenue of an auction $A$ can be computed as a linear function of $X_{ij}^A(\dist,S_1,\ldots, S_n)$ and $U_{ij}^A(\dist,S_1,\ldots, S_n)$ (because revenue $=$ welfare $-$ utility). Therefore, Theorem~\ref{thm:BMSWfeasible} enables a simple LP formulation to find the optimal auction for clever buyers. 

We consider two natural attempts to generalize Theorem~\ref{thm:BMSWfeasible}, and show that both hold only in the `only if' direction. The reason the BMSW constraints don't work verbatim for multiple buyers is that the feasibility constraints are wrong: it is not feasible to (for instance) have each buyer win the item with probability $1$ every round. Indeed, there is only one copy of the item, implying (for instance) that $n\sum_i q_i \cdot x(w_i) \leq 1$, but also stronger conditions. These conditions are known as \emph{Border's constraints} from Theorem~\ref{thm:Border}~\cite{Border91}. 

\begin{proposition}\label{prop:feasible} A tuple $(x,y,u)$ is $n$-buyer auction feasible only if it satisfies the \emph{$n$-buyer BMSW constraints} below. A tuple $(x,u)$ is $n$-buyer auction feasible only if it satisfies the \emph{reduced $n$-buyer BMSW constraints}.\footnote{In fact, this claim holds when replacing mean-based clever buyers with just clever buyers, just like the `only if' part of Theorem~\ref{thm:BMSWfeasible}.}
\begin{align*}
    &\textbf{$n$-buyer BMSW Constraints} &\textbf{Reduced $n$-buyer BMSW Constraints} \\
    & u_i \geq (w_i - w_j) \cdot y_j, \enspace \forall i,j \in [m] : i > j, 
&u_i \geq (w_i - w_j) \cdot x_j, \enspace \forall i,j \in [m]: i > j,\\ 
    & y_i \geq x_i, \enspace \forall i \in [m],\\
    & u_i \geq 0, \enspace \forall i \in [m],
&u_i \geq 0, \enspace \forall i \in [m],\\
    & \vec{x} \text{ satisfies Border's constraints for $n$ buyers,}
&\vec{x} \text{ satisfies Border's constraints for $n$ buyers,}\\
    & \vec{x},\vec{y} \text{ monotone.}
&\vec{x} \text{ monotone.}\\
\end{align*}
\end{proposition}
\begin{proof}
First, we prove the proposition for tuples $(x,y,u)$. Assume that a tuple $(x,y,u)$ is $n$-buyer auction feasible, and consider the auction that matches it. 

We know that $\vec{y}$ must be monotone, because higher arms must award the item more often. Similarly, we know that $\vec{x}$ must be monotone, because whenever two values pull different arms, the higher value must pull a higher arm (again because of monotonicity of the auction). 

We know that $\vec{x}$ must satisfy Border's condition for $n$ buyers, because $\vec{x}$ determines the probability that the buyer \emph{actually gets the item} using their strategy throughout the course of the auction. In other words, the auction ``have the buyers report a strategy for the repeated auction. Run those strategies and pick a uniformly random round. Award the item according to that round.'' is a single-item auction which gives a buyer with value $w_j$ the item with probability exactly $x_j$ (over the randomness in other buyers drawing their values i.i.d.~from $\dist$). Because this is a feasible, monotone auction, it must satisfy Border's constraints.\footnote{Note that we \emph{cannot} say the same about $\vec{y}$ --- an undesirable arm can \emph{promise} the item with probability $1$, as long as no one actually pulls it. As an example, consider running a first-price auction every round. No one would pull arm $w_m$, but doing so would indeed win the item with probability $1$.}

Clearly, $u_i\geq 0$ because buyers can always pull the null arm every round to get $0$ payoff, and they have no regret.

To see that $y_i \geq x_i$, recall that a clever buyer \emph{never overbids}. This means that they always pull an arm no higher than $w_i$. Because the auction is monotone, this means that during every round, arm $i$ gives the item with (weakly) higher probability than the arm that a buyer with value $w_i$ pulls. Because this holds every round, it clearly holds on average as well.

Finally, to see that $u_i \geq (w_i - w_j) \cdot y_j$, observe that buyers with value $w_i$ must not regret pulling arm $w_j$. Recall that arm $w_j$ can charge at most $w_j y_j$ on average, and awards the item on average with probability $y_j$. So the buyer's utility for this arm is at least $(w_i - w_j)$ conditioned on winning the item, and the item is won with probability $y_j$ (meaning their average utility is at least $(w_i -w_j)y_j$). Because the buyer is no-regret, $u_i$ must exceed this. 

This completes the proof for tuples $(x,y,u)$. Observe that we have not used the fact that the buyers are mean-based, just that they are clever. Moreover, we have only used the fact that buyers are clever to conclude $y_i \geq x_i$ (all other constraints hold just by the fact that the buyer guarantees no-regret).

To see the proposition for tuples $(x,u)$, assume that there exists a $y$ such that $(x,y,u)$ is $n$-buyer auction feasible. Then $y_i \geq x_i$, and therefore $u_i \geq (w_i - w_j) \cdot y_j \geq (w_i-w_j) \cdot x_j$, as desired. The other constraints are copied immediately from the fact that $(x,y,u)$ is $n$-buyer auction feasible. This completes the proof for tuples $(x,u)$.
\end{proof}

We next turn to see whether the other direction holds, as in Theorem~\ref{thm:BMSWfeasible} for the single-buyer case. If it did, then we could again write a linear program to find the optimal $n$-buyer feasible auction, because the expected revenue can be written as a function of $(x,u)$. However, we provide an example showing that this extension is \emph{false}. 

\begin{theorem}\label{thm:counterexample} There exist $(x,y,u)$ that satisfy the $n$-buyer BMSW Constraints but are not $n$-buyer auction feasible, and $(x,u)$ that satisfy the Reduced $n$-buyer BMSW Constraints but are not $n$-buyer auction feasible. 
\end{theorem}

We provide a proof of Theorem~\ref{thm:counterexample} in Appendix~\ref{app:UB}, and overview the key properties of our construction for $(x,u)$ here, and highlight why such a construction cannot arise when $n=1$ (aside from the fact that it contradicts Theorem~\ref{thm:BMSWfeasible}). Our construction for $(x,y,u)$ simply takes this construction and sets $x=y$. 

One aspect of our construction is that $x$ is an \emph{extreme point} of the Border polytope. That is, it cannot be written as a convex combination of $\geq 2$ distinct points that also satisfy the Border constraints. This immediately implies that any auction which matches $(x,u)$ \emph{must} have buyers with value $w_j$ receive the item with probability $x_j$ in \emph{every single round} (except for $o(T)$). This is because $x$ is a convex combination of the allocation probabilities in each fixed round, and the allocation probabilities in each fixed round must be in the Border polytope. Because ${x}$ is an extreme point, all the allocations during each round must be $x$ itself.

The second key aspect of our construction is that we will have the buyer with value $w_4$ strictly prefer arm $w_2$ to all other arms, and the buyer with value $w_3$ strictly prefer arm $w_1$ to all other arms. We defer to the full proof an explanation for why this aspect drives our example, but we elaborate here on why these two key properties can't coexist when $n=1$. When $n=1$, the extreme points of the Border polytope are exceptionally simple: there exists a $j$ such that $x_\ell = 1$ for all $\ell \geq j$, and $x_\ell = 0$ for all $\ell < j$. For any $x$ with this property, it cannot \emph{also} be that $w_4$ strictly prefers $w_2$ to all other arms, while $w_3$ strictly prefers $w_1$. Indeed, if $j =1$, then all buyers weakly prefer $w_1$ to all other arms. If instead $j\geq 2$, then a buyer with value $w_3$ weakly prefers arm $w_2$ to arm $w_1$. But with two buyers, the Border polytope is rich enough to have both properties simultaneously.

We defer further details to Appendix~\ref{app:UB}, and quickly conclude with the main point of this section: there is a natural extension of the BMSW conditions to $n$ buyers, which are necessary but \emph{not sufficient} to characterize tuples which are $n$-buyer auction feasible.

\subsection{Uniform Auctions with Declining Reserves}\label{sec:UPYB}
In this section, we consider the following possibility: although $n$-buyer BMSW constraints don't characterize the $n$-buyer feasible auctions, it is conceivable (although perhaps unlikely) that the Linear Programming solution (optimizing expected revenue subject to $n$-buyer BMSW constraints) happens to always yield an $n$-buyer feasible auction. The reason this is a priori possible is because the objective function and BMSW constraints are related: the objective function depends on $\dist$, and so do the $n$-buyer Border constraints (this is another way in which $n$-buyer and $1$-buyer auctions differ: $1$-buyer Border constraints don't depend on $\dist$). 

For the single-buyer case,~\cite{BravermanMSW18} shows that not only is every tuple satisfying the BMSW constraints $1$-buyer auction feasible, but the auction witnessing this is particularly simple. First, whenever the buyer gets the item, they pay their bid (and in each round, each arm gives the item with probability $0$ or $1$). Second, the minimum winning bid is declining over time. We generalize both definitions below to multiple buyers, and show a connection between these auctions and certain types of tuples which satisfy the $n$-buyer BMSW conditions.

\begin{definition}[Pay-your-bid] A repeated auction is \textit{pay-your-bid} if $p_{i,t}(\vec{b}) = b_i \cdot a_{i,t}(\vec{b})$ for all $i,t$.
\end{definition}

\begin{definition}[Uniform Auction with Declining Reserve]
A repeated auction is a \emph{uniform auction with declining reserve} when: (a) there exists a reserve $r_t$ for every round $t$ which is monotonically decreasing in $t$, and (b) in each round the item is awarded to a uniformly random buyer among those with $b_{it} \geq r_t$. 
\end{definition}

\begin{definition}[Correspondence]
We call $(x, y, u)$ the \emph{corresponding tuple} of repeated auction $A$ if for $0$-mean-based strategies $S_1,\ldots, S_n$ and all $i, j$: $X_{ij}^A(\dist, S_1, \ldots, S_n) = x_j;$ $ Y_{ij}^A(\dist, S_1, \ldots, S_n) = y_j;
U_{ij}^A(\dist, S_1, \ldots, S_n) = u_j$.\footnote{Observe that this is always well-defined, as the unique $0$-mean-based strategy is Follow-the-leader.}
\end{definition}

In this language,~\cite{BravermanMSW18} shows that when $n=1$, every tuple which satisfies the BMSW conditions can be implemented as a pay-your-bid uniform auction with declining reserve (and this establishes the `if' direction of Theorem~\ref{thm:BMSWfeasible}). Due to Theorem~\ref{thm:counterexample}, this claim clearly cannot extend to $n> 1$. However, we show that certain kinds of natural tuples can all be implemented as pay-your-bid uniform auctions with declining reserves.

\begin{theorem}\label{thm:uniform} Consider any repeated auction $A$ and its corresponding tuple $(x,y,u)$. If $(x,y,u)$ satisfies the $n$-buyer BMSW constraints, and $x=y$, and $A$ is pay-your-bid, then $A$ is a uniform auction with declining reserve.
\end{theorem}

Before proving Theorem~\ref{thm:uniform}, we highlight why these are natural tuples to consider. First, recall the BMSW constraint $u_i \geq (w_i - w_j) \cdot y_j$. In order for this constraint to possibly be tight, it must be that arm $w_j$ charges $w_j$ in every round (i.e. the auction is pay-your-bid). There is also a BMSW constraint $y_i \geq x_i$. Clearly, this constraint is only tight when $y_i = x_i$. So such auctions naturally capture certain extreme points of the $n$-buyer BMSW polytope. Similarly, for the reduced $n$-buyer BMSW constraints, observe that in order for $u_i \geq (w_i - w_j)\cdot x_j$ to be tight it must \emph{both} be the case that the auction is pay-your-bid and that $y_j = x_j$. So again, these conditions form a natural subset of extreme points of the reduced $n$-buyer BMSW polytope.

\begin{proof}[Proof of Theorem~\ref{thm:uniform}]
We first prove that $A$ has the following ``uniform allocation'' property: during every round, let $S$ denote the set of arms which award the item with non-zero probability. Then the item is awarded uniformly at random to a buyer who pulled an arm in $S$. To see this, consider what happens when $0$-mean based buyers participate in auction $A$. Recall first that because the auction is pay-your-bid, that a buyer with value $w_j$ gets $0$ cumulative utility from bidding $w_j$. This means that as long as any arm $w_\ell$, for $\ell < j$, has historically awarded the item in $>0$ rounds, a $0$-mean-based buyer will never pull arm $j$. This means that there is at most one value $w_j$ who would pull their own arm and have a non-zero probability of winning the item. 

So let $j^*$ denote the lowest-indexed value that pulls with non-negligible probability an arm that wins the item with non-zero probability. Then all values $w_j$, for $j > j^*$ also pull an arm that wins the item with non-zero probability. And moreover, all such $w_j$ pull an arm $w_\ell$ for $\ell < j$. Now, we can use the fact that $x=y$ to conclude that the bid $w_j$ and $w_\ell$ must win the item with identical probability. To see this, observe that in every round, arm $w_j$ wins the item with probability at least as much as the buyer with value $w_j$ (because buyers never overbid). Therefore, in order to possibly have $x=y$, it must be that these two probabilities are the same in every round. But now, we can chain these claims together: value $w_{j^*+1}$ must pull arm $w_{j^*}$, so they must award the item with the same probability. In general, for all $j > j^*$, value $w_j$ must pull an arm $w_\ell$ for $\ell \in [j^*,j-1]$. Therefore, a proof by induction concludes that all values $w_j$ for $j > j^*$ must pull an arm that gives them the same winning probability as $w_{j^*}$, and $A$ must have the uniform allocation property. 

Now that $A$ has the uniform allocation property, there is a well-defined reserve in each round (the minimum bid that can win the item). The remaining task is to show the reserve decreases over time. 

Consider a round $s$ with reserve $r_s$, and assume for contradiction that we have at least one round with a reserve $< r_s$. Then consider when the buyer has value ${r_s}$ during round $s$. Because in some previous round, the reserve was $< r_s$, we know that the buyer with value ${r_s}$ has cumulative utility $>0$ for the arm immediately below $r_s$, and therefore does not pulls their own arm (which gives $0$ cumulative utility because the auction is pay-your-bid). Therefore, during round $s$, a buyer with value ${r_s}$ must pull an arm $w_j < r_s$. Because $x=y$, arm $w_j$ must also be above the reserve, contradicting that $w_j < r_s$ but $r_s$ is the reserve.

The first half of the proof establishes that the auction must be a uniform auction, and the second half establishes that the reserve must be monotonically declining.
\end{proof}

With Theorem~\ref{thm:uniform} in mind, another possible avenue towards characterizing optimal $n$-buyer feasible auctions would be through pay-your-bid uniform auctions with declining reserves. To this end, we first show that the optimal pay-your-bid uniform auction with declining reserve can be found by a linear program. However, we also show that examples exist where the optimal $n$-buyer feasible auction strictly outperforms the best pay-your-bid uniform auction with declining reserve. Recall again that Theorem~\ref{thm:BMSWfeasible} establishes that the best $1$-buyer feasible auction is always a pay-your-bid uniform auction with declining reserve. The proof is in Appendix~\ref{app:UPYB}.

\begin{theorem}\label{thm:uniform2} The optimal{\footnote{By optimal we mean the auction achieves the best revenue if all buyers run $0$-regret algorithms. It is easy to see that when buyers have $\gamma$-regret, the revenue is within $O(n \gamma)$ of the revenue when buyers have $0$-regret.}}  pay-your-bid uniform auction with declining reserve can be found by a linear program of size $\poly(m)$. However, there exist $2$-buyer instances where the optimal $2$-buyer feasible auction strictly outperforms the best pay-your-bid uniform auction with declining reserve.
\end{theorem}

\subsection{Non-convexity of n-buyer Feasible Auctions}\label{sec:nonconvex}
Finally, we consider the possibility that while the $n$-buyer BMSW constraints don't capture the space of $n$-buyer feasible auctions, perhaps some other compact, convex space does. This too is not the case, as we show that the space of $n$-buyer feasible triples is non-convex (subject to one technical restriction).

\begin{theorem}\label{thm:nonconvex} Let $P$ denote the set of all $(x,y,u)$ that are $n$-buyer feasible auctions \emph{where the bid space is equal to the support of $\dist$}. Then $P$ is not necessarily convex, even when $n=2$. 
\end{theorem}

We prove Theorem~\ref{thm:nonconvex} in Appendix~\ref{app:nonconvex} by construction, and defer technical details. We note that Theorem~\ref{thm:nonconvex} provides an alternative proof that the $n$-buyer BMSW constraints do not capture $n$-buyer feasible auctions (if they did, the space would be convex). It does not imply anything regarding the reduced $n$-buyer BMSW constraints.

\section{Conclusion}\label{sec:conclusion}
We study the repeated sale of a single item to multiple mean-based no-regret buyers. Our main result designs an auction which extracts the full expected value as revenue from the buyers (Section~\ref{sec:FSE}). Our other results identify several formal barriers to extending the analysis of~\cite{BravermanMSW18} to multiple clever, mean-based buyers (Section~\ref{sec:UB}). 

While our work resolves the multi-buyer study of mean-based buyers, there are many natural directions in the study of clever mean-based buyers. For example:

\begin{itemize}
\item Our work shows that the BMSW constraints are necessary but not sufficient for clever mean-based buyers. Are there other classes of clever no-regret learners for which variants of the BMSW constraints are necessary and sufficient? Or at least for which the space of implementable auctions is convex?
\item We show that natural generalizations of the single-buyer BMSW results to multiple clever mean-based buyers fail. Is there a significantly different approach to characterize optimal auctions for clever mean-based buyers?
\item Finally, can simple auctions approximate the optimum? Note that even for a single buyer, the BMSW auction is unboundedly better than any truthful auction. But perhaps the pay-your-bid uniform auction with declining reserve is a natural starting point (as it is optimal for a single buyer).
\end{itemize}

In addition,~\cite{BravermanMSW18} provides an online learning algorithm that is \emph{not} mean-based, where the seller's optimal achievable revenue against such learners is simply that of Myerson's optimal auction per-round. Given our main result, it is also important to further study \emph{classes} of online learning algorithms which are not mean-based, and determine which classes also allow for full surplus extraction.

\bibliographystyle{alpha}
\bibliography{MasterBib}

\appendix 
\section{Example of the FSE Auction} \label{app:example}
\begin{figure}[h!]
    \centering
    \tikzset{every picture/.style={line width=0.75pt}} 

\begin{tikzpicture}[x=0.7pt,y=0.7pt,yscale=-1,xscale=1]

\draw   (48,461.33) -- (605,461.33) -- (605,486) -- (48,486) -- cycle ;
\draw  [fill={rgb, 255:red, 255; green, 255; blue, 255 }  ,fill opacity=1 ] (84,461.33) -- (120,461.33) -- (120,486) -- (84,486) -- cycle ;
\draw  [fill={rgb, 255:red, 245; green, 166; blue, 35 }  ,fill opacity=1 ] (48,461.33) -- (84,461.33) -- (84,486) -- (48,486) -- cycle ;
\draw   (50,490.33) .. controls (50,494.72) and (52.2,496.92) .. (56.59,496.92) -- (56.59,496.92) .. controls (62.86,496.92) and (66,499.12) .. (66,503.51) .. controls (66,499.12) and (69.14,496.92) .. (75.41,496.92)(72.59,496.92) -- (75.41,496.92) .. controls (79.8,496.92) and (82,494.72) .. (82,490.33) ;
\draw   (86,489.33) .. controls (86,494) and (88.33,496.33) .. (93,496.33) -- (334.5,496.33) .. controls (341.17,496.33) and (344.5,498.66) .. (344.5,503.33) .. controls (344.5,498.66) and (347.83,496.33) .. (354.5,496.33)(351.5,496.33) -- (596,496.33) .. controls (600.67,496.33) and (603,494) .. (603,489.33) ;
\draw   (120,461.33) -- (156,461.33) -- (156,486) -- (120,486) -- cycle ;
\draw   (156,461.33) -- (192,461.33) -- (192,486) -- (156,486) -- cycle ;
\draw   (192,461.33) -- (228,461.33) -- (228,486) -- (192,486) -- cycle ;
\draw   (228,461.33) -- (264,461.33) -- (264,486) -- (228,486) -- cycle ;
\draw   (569,461.33) -- (605,461.33) -- (605,486) -- (569,486) -- cycle ;
\draw    (64,460.33) -- (71.15,433.43) ;
\draw [shift={(71.67,431.5)}, rotate = 464.89] [color={rgb, 255:red, 0; green, 0; blue, 0 }  ][line width=0.75]    (10.93,-3.29) .. controls (6.95,-1.4) and (3.31,-0.3) .. (0,0) .. controls (3.31,0.3) and (6.95,1.4) .. (10.93,3.29)   ;

\draw (347,462) node [anchor=north west][inner sep=0.75pt]   [align=left] {$\displaystyle \cdots $};
\draw (45,508) node [anchor=north west][inner sep=0.75pt]   [align=left] {active arms};
\draw (340,501) node [anchor=north west][inner sep=0.75pt]   [align=left] {dormant arms};
\draw (57,464) node [anchor=north west][inner sep=0.75pt]  [font=\footnotesize]  {$b_{P}$};
\draw (91,464.33) node [anchor=north west][inner sep=0.75pt]  [font=\footnotesize]  {$b_{P-1}$};
\draw (125,464.33) node [anchor=north west][inner sep=0.75pt]  [font=\footnotesize]  {$b_{P-2}$};
\draw (161,464.33) node [anchor=north west][inner sep=0.75pt]  [font=\footnotesize]  {$b_{P-3}$};
\draw (197,464.33) node [anchor=north west][inner sep=0.75pt]  [font=\footnotesize]  {$b_{P-4}$};
\draw (233,465.33) node [anchor=north west][inner sep=0.75pt]  [font=\footnotesize]  {$b_{P-5}$};
\draw (578,464.33) node [anchor=north west][inner sep=0.75pt]  [font=\footnotesize]  {$b_{1}$};
\draw (42,410) node [anchor=north west][inner sep=0.75pt]   [align=left] {pulled by $\displaystyle w_{1} ,\ w_{2} ,\ w_{3} ,\ w_{4}$};
\draw (42,392) node [anchor=north west][inner sep=0.75pt]   [align=left] {submit bid $\displaystyle w_{1}$};

\end{tikzpicture}
    \caption{support size $m = 4$, set up phase ($\phaseVar = 1$)} 
    \label{fig:auctionIter1}
\end{figure}
\begin{figure}[h!]
    \centering
    \tikzset{every picture/.style={line width=0.75pt}} 

\begin{tikzpicture}[x=0.7pt,y=0.7pt,yscale=-1,xscale=1]

\draw   (48,619.33) -- (605,619.33) -- (605,644) -- (48,644) -- cycle ;
\draw  [fill={rgb, 255:red, 245; green, 166; blue, 35 }  ,fill opacity=1 ] (84,619.33) -- (120,619.33) -- (120,644) -- (84,644) -- cycle ;
\draw  [fill={rgb, 255:red, 245; green, 166; blue, 35 }  ,fill opacity=1 ] (48,619.33) -- (84,619.33) -- (84,644) -- (48,644) -- cycle ;
\draw   (50,648.33) .. controls (50,652.72) and (52.2,654.92) .. (56.59,654.92) -- (56.59,654.92) .. controls (62.86,654.92) and (66,657.12) .. (66,661.51) .. controls (66,657.12) and (69.14,654.92) .. (75.41,654.92)(72.59,654.92) -- (75.41,654.92) .. controls (79.8,654.92) and (82,652.72) .. (82,648.33) ;
\draw   (86,647.33) .. controls (86,652) and (88.33,654.33) .. (93,654.33) -- (334.5,654.33) .. controls (341.17,654.33) and (344.5,656.66) .. (344.5,661.33) .. controls (344.5,656.66) and (347.83,654.33) .. (354.5,654.33)(351.5,654.33) -- (596,654.33) .. controls (600.67,654.33) and (603,652) .. (603,647.33) ;
\draw   (120,619.33) -- (156,619.33) -- (156,644) -- (120,644) -- cycle ;
\draw   (156,619.33) -- (192,619.33) -- (192,644) -- (156,644) -- cycle ;
\draw   (192,619.33) -- (228,619.33) -- (228,644) -- (192,644) -- cycle ;
\draw   (228,619.33) -- (264,619.33) -- (264,644) -- (228,644) -- cycle ;
\draw   (569,619.33) -- (605,619.33) -- (605,644) -- (569,644) -- cycle ;
\draw    (64,618.33) -- (71.15,591.43) ;
\draw [shift={(71.67,589.5)}, rotate = 464.89] [color={rgb, 255:red, 0; green, 0; blue, 0 }  ][line width=0.75]    (10.93,-3.29) .. controls (6.95,-1.4) and (3.31,-0.3) .. (0,0) .. controls (3.31,0.3) and (6.95,1.4) .. (10.93,3.29)   ;
\draw    (104.67,619.5) -- (230.71,593.9) ;
\draw [shift={(232.67,593.5)}, rotate = 528.52] [color={rgb, 255:red, 0; green, 0; blue, 0 }  ][line width=0.75]    (10.93,-3.29) .. controls (6.95,-1.4) and (3.31,-0.3) .. (0,0) .. controls (3.31,0.3) and (6.95,1.4) .. (10.93,3.29)   ;

\draw (347,620) node [anchor=north west][inner sep=0.75pt]   [align=left] {$\displaystyle \cdots $};
\draw (45,666) node [anchor=north west][inner sep=0.75pt]   [align=left] {active arms};
\draw (340,659) node [anchor=north west][inner sep=0.75pt]   [align=left] {dormant arms};
\draw (57,622) node [anchor=north west][inner sep=0.75pt]  [font=\footnotesize]  {$b_{P}$};
\draw (91,622.33) node [anchor=north west][inner sep=0.75pt]  [font=\footnotesize]  {$b_{P-1}$};
\draw (125,622.33) node [anchor=north west][inner sep=0.75pt]  [font=\footnotesize]  {$b_{P-2}$};
\draw (161,622.33) node [anchor=north west][inner sep=0.75pt]  [font=\footnotesize]  {$b_{P-3}$};
\draw (197,622.33) node [anchor=north west][inner sep=0.75pt]  [font=\footnotesize]  {$b_{P-4}$};
\draw (233,623.33) node [anchor=north west][inner sep=0.75pt]  [font=\footnotesize]  {$b_{P-5}$};
\draw (578,622.33) node [anchor=north west][inner sep=0.75pt]  [font=\footnotesize]  {$b_{1}$};
\draw (42,568) node [anchor=north west][inner sep=0.75pt]   [align=left] {pulled by $\displaystyle w_{2} ,\ w_{3} ,\ w_{4}$};
\draw (42,550) node [anchor=north west][inner sep=0.75pt]   [align=left] {submit bid $\displaystyle w_{2}$};
\draw (204,566) node [anchor=north west][inner sep=0.75pt]   [align=left] {pulled by $\displaystyle w_{1}$};
\draw (204,550) node [anchor=north west][inner sep=0.75pt]   [align=left] {submit bid $\displaystyle w_{1}$};

\end{tikzpicture}
    \caption{support size $m = 4$, set up phase ($\phaseVar = 2$)} 
    \label{fig:auctionIter2}
\end{figure}
\begin{figure}[h!]
    \centering
    \tikzset{every picture/.style={line width=0.75pt}} 

\begin{tikzpicture}[x=0.7pt,y=0.7pt,yscale=-1,xscale=1]

\draw   (46,95.33) -- (603,95.33) -- (603,120) -- (46,120) -- cycle ;
\draw  [fill={rgb, 255:red, 245; green, 166; blue, 35 }  ,fill opacity=1 ] (82,95.33) -- (118,95.33) -- (118,120) -- (82,120) -- cycle ;
\draw  [fill={rgb, 255:red, 80; green, 227; blue, 194 }  ,fill opacity=1 ] (46,95.33) -- (82,95.33) -- (82,120) -- (46,120) -- cycle ;
\draw   (83,125.33) .. controls (83,130) and (85.33,132.33) .. (90,132.33) -- (142.59,132.33) .. controls (149.26,132.33) and (152.59,134.66) .. (152.59,139.33) .. controls (152.59,134.66) and (155.92,132.33) .. (162.59,132.33)(159.59,132.33) -- (217,132.33) .. controls (221.67,132.33) and (224,130) .. (224,125.33) ;
\draw   (228,125.33) .. controls (228.03,130) and (230.37,132.32) .. (235.04,132.3) -- (404.54,131.38) .. controls (411.21,131.34) and (414.55,133.65) .. (414.58,138.32) .. controls (414.55,133.65) and (417.87,131.3) .. (424.54,131.27)(421.54,131.29) -- (594.04,130.35) .. controls (598.71,130.33) and (601.03,127.99) .. (601,123.32) ;
\draw   (81,91.33) .. controls (81,86.66) and (78.67,84.33) .. (74,84.33) -- (74,84.33) .. controls (67.33,84.33) and (64,82) .. (64,77.33) .. controls (64,82) and (60.67,84.33) .. (54,84.33)(57,84.33) -- (54,84.33) .. controls (49.33,84.33) and (47,86.66) .. (47,91.33) ;
\draw  [fill={rgb, 255:red, 245; green, 166; blue, 35 }  ,fill opacity=1 ] (118,95.33) -- (154,95.33) -- (154,120) -- (118,120) -- cycle ;
\draw  [fill={rgb, 255:red, 245; green, 166; blue, 35 }  ,fill opacity=1 ] (154,95.33) -- (190,95.33) -- (190,120) -- (154,120) -- cycle ;
\draw  [fill={rgb, 255:red, 245; green, 166; blue, 35 }  ,fill opacity=1 ] (190,95.33) -- (226,95.33) -- (226,120) -- (190,120) -- cycle ;
\draw   (226,95.33) -- (262,95.33) -- (262,120) -- (226,120) -- cycle ;
\draw   (567,95.33) -- (603,95.33) -- (603,120) -- (567,120) -- cycle ;
\draw    (101,95.33) -- (158.23,65.26) ;
\draw [shift={(160,64.33)}, rotate = 512.28] [color={rgb, 255:red, 0; green, 0; blue, 0 }  ][line width=0.75]    (10.93,-3.29) .. controls (6.95,-1.4) and (3.31,-0.3) .. (0,0) .. controls (3.31,0.3) and (6.95,1.4) .. (10.93,3.29)   ;
\draw    (211,95.33) -- (268.23,65.26) ;
\draw [shift={(270,64.33)}, rotate = 512.28] [color={rgb, 255:red, 0; green, 0; blue, 0 }  ][line width=0.75]    (10.93,-3.29) .. controls (6.95,-1.4) and (3.31,-0.3) .. (0,0) .. controls (3.31,0.3) and (6.95,1.4) .. (10.93,3.29)   ;

\draw (345,96) node [anchor=north west][inner sep=0.75pt]   [align=left] {$\displaystyle \cdots $};
\draw (112,141) node [anchor=north west][inner sep=0.75pt]   [align=left] {active arms};
\draw (338,135) node [anchor=north west][inner sep=0.75pt]   [align=left] {dormant arms};
\draw (43,56) node [anchor=north west][inner sep=0.75pt]   [align=left] {retired arms};
\draw (55,98) node [anchor=north west][inner sep=0.75pt]  [font=\footnotesize]  {$b_{P}$};
\draw (88,98.33) node [anchor=north west][inner sep=0.75pt]  [font=\footnotesize]  {$b_{P-1}$};
\draw (123,98.33) node [anchor=north west][inner sep=0.75pt]  [font=\footnotesize]  {$b_{P-2}$};
\draw (159,98.33) node [anchor=north west][inner sep=0.75pt]  [font=\footnotesize]  {$b_{P-3}$};
\draw (195,98.33) node [anchor=north west][inner sep=0.75pt]  [font=\footnotesize]  {$b_{P-4}$};
\draw (231,99.33) node [anchor=north west][inner sep=0.75pt]  [font=\footnotesize]  {$b_{P-5}$};
\draw (576,98.33) node [anchor=north west][inner sep=0.75pt]  [font=\footnotesize]  {$b_{1}$};
\draw (136,41) node [anchor=north west][inner sep=0.75pt]   [align=left] {pulled by $\displaystyle w_{4}$};
\draw (240,42) node [anchor=north west][inner sep=0.75pt]   [align=left] {pulled by $\displaystyle w_{1}$};
\draw (136,23) node [anchor=north west][inner sep=0.75pt]   [align=left] {submit bid $\displaystyle w_{4}$};
\draw (241,23) node [anchor=north west][inner sep=0.75pt]   [align=left] {submit bid $\displaystyle w_{1}$};

\end{tikzpicture}
    \caption{support size $m = 4$, main phase ($\phaseVar = 5$)} 
    \label{fig:auctionIter5}
\end{figure}
\begin{figure}[h!]
    \centering
    \tikzset{every picture/.style={line width=0.75pt}} 

\begin{tikzpicture}[x=0.7pt,y=0.7pt,yscale=-1,xscale=1]

\draw   (42,281.33) -- (599,281.33) -- (599,306) -- (42,306) -- cycle ;
\draw  [fill={rgb, 255:red, 80; green, 227; blue, 194 }  ,fill opacity=1 ] (78,281.33) -- (114,281.33) -- (114,306) -- (78,306) -- cycle ;
\draw  [fill={rgb, 255:red, 80; green, 227; blue, 194 }  ,fill opacity=1 ] (42,281.33) -- (78,281.33) -- (78,306) -- (42,306) -- cycle ;
\draw   (114,310.33) .. controls (114,315) and (116.33,317.33) .. (121,317.33) -- (173.59,317.33) .. controls (180.26,317.33) and (183.59,319.66) .. (183.59,324.33) .. controls (183.59,319.66) and (186.92,317.33) .. (193.59,317.33)(190.59,317.33) -- (248,317.33) .. controls (252.67,317.33) and (255,315) .. (255,310.33) ;
\draw   (260,310.33) .. controls (260.01,315) and (262.35,317.32) .. (267.02,317.31) -- (419.02,316.86) .. controls (425.69,316.84) and (429.03,319.16) .. (429.04,323.83) .. controls (429.03,319.16) and (432.35,316.82) .. (439.02,316.8)(436.02,316.81) -- (591.02,316.35) .. controls (595.69,316.34) and (598.01,314) .. (598,309.33) ;
\draw   (112,277.33) .. controls (112,272.66) and (109.67,270.33) .. (105,270.33) -- (85.91,270.33) .. controls (79.24,270.33) and (75.91,268) .. (75.91,263.33) .. controls (75.91,268) and (72.58,270.33) .. (65.91,270.33)(68.91,270.33) -- (50,270.33) .. controls (45.33,270.33) and (43,272.66) .. (43,277.33) ;
\draw  [fill={rgb, 255:red, 245; green, 166; blue, 35 }  ,fill opacity=1 ] (114,281.33) -- (150,281.33) -- (150,306) -- (114,306) -- cycle ;
\draw  [fill={rgb, 255:red, 245; green, 166; blue, 35 }  ,fill opacity=1 ] (150,281.33) -- (186,281.33) -- (186,306) -- (150,306) -- cycle ;
\draw  [fill={rgb, 255:red, 245; green, 166; blue, 35 }  ,fill opacity=1 ] (186,281.33) -- (222,281.33) -- (222,306) -- (186,306) -- cycle ;
\draw  [fill={rgb, 255:red, 245; green, 166; blue, 35 }  ,fill opacity=1 ] (222,281.33) -- (258,281.33) -- (258,306) -- (222,306) -- cycle ;
\draw   (563,281.33) -- (599,281.33) -- (599,306) -- (563,306) -- cycle ;
\draw   (258,281.33) -- (294,281.33) -- (294,306) -- (258,306) -- cycle ;
\draw    (132,280.33) -- (189.23,250.26) ;
\draw [shift={(191,249.33)}, rotate = 512.28] [color={rgb, 255:red, 0; green, 0; blue, 0 }  ][line width=0.75]    (10.93,-3.29) .. controls (6.95,-1.4) and (3.31,-0.3) .. (0,0) .. controls (3.31,0.3) and (6.95,1.4) .. (10.93,3.29)   ;
\draw    (239,281.33) -- (296.23,251.26) ;
\draw [shift={(298,250.33)}, rotate = 512.28] [color={rgb, 255:red, 0; green, 0; blue, 0 }  ][line width=0.75]    (10.93,-3.29) .. controls (6.95,-1.4) and (3.31,-0.3) .. (0,0) .. controls (3.31,0.3) and (6.95,1.4) .. (10.93,3.29)   ;

\draw (341,282) node [anchor=north west][inner sep=0.75pt]   [align=left] {$\displaystyle \cdots $};
\draw (145,321) node [anchor=north west][inner sep=0.75pt]   [align=left] {active arms};
\draw (379,322) node [anchor=north west][inner sep=0.75pt]   [align=left] {dormant arms};
\draw (39,242) node [anchor=north west][inner sep=0.75pt]   [align=left] {retired arms};
\draw (51,284) node [anchor=north west][inner sep=0.75pt]  [font=\footnotesize]  {$b_{P}$};
\draw (84,284.33) node [anchor=north west][inner sep=0.75pt]  [font=\footnotesize]  {$b_{P-1}$};
\draw (119,284.33) node [anchor=north west][inner sep=0.75pt]  [font=\footnotesize]  {$b_{P-2}$};
\draw (155,284.33) node [anchor=north west][inner sep=0.75pt]  [font=\footnotesize]  {$b_{P-3}$};
\draw (191,284.33) node [anchor=north west][inner sep=0.75pt]  [font=\footnotesize]  {$b_{P-4}$};
\draw (227,285.33) node [anchor=north west][inner sep=0.75pt]  [font=\footnotesize]  {$b_{P-5}$};
\draw (572,284.33) node [anchor=north west][inner sep=0.75pt]  [font=\footnotesize]  {$b_{1}$};
\draw (263,283.33) node [anchor=north west][inner sep=0.75pt]  [font=\footnotesize]  {$b_{P-6}$};
\draw (167,226) node [anchor=north west][inner sep=0.75pt]   [align=left] {pulled by $\displaystyle w_{4}$};
\draw (268,227) node [anchor=north west][inner sep=0.75pt]   [align=left] {pulled by $\displaystyle w_{1}$};
\draw (269,210) node [anchor=north west][inner sep=0.75pt]   [align=left] {submit bid $\displaystyle w_{1}$};
\draw (167,208) node [anchor=north west][inner sep=0.75pt]   [align=left] {submit bid $\displaystyle w_{4}$};

\end{tikzpicture}
    \caption{support size $m = 4$, main phase ($\phaseVar = 6$)} 
    \label{fig:auctionIter6}
\end{figure}

In this section we describe the behavior of mean based no regret buyers in the FSE auction for a concrete example. 

Let us consider an example where there are two buyers participating in our full surplus extraction auction, and the distribution $\D$ is a uniform distribution over $\{w_1, w_2, w_3, w_4\}$. Here $m$, the support size of $\D$, is equal to $4$.

The first $3$ phases are the set up phase, where the purpose is to make buyer pull a distinct arm when they have each distinct value after the set up phases. For instance, in the first phase, only arm $b_{P}$ is active, and all other arms are dormant (which means these arms will not give the item and will also not charge anything). We will see that both buyers, no matter their value, will choose to pull $b_{P}$ in this phase. However, starting from phase $2$, the buyers' choice of arm will become more and more affected by their value. 

In each round of phase $1$, if a buyer pulls the arm $b_{P}$, then a bid of $w_1$ is entered into the second price auction. Assuming that when a buyer gets $0$ utility, they tie break favoring getting the item, no matter what the two buyer's values are, they will both choose to pull arm $b_{P}$, so $S$ the set of buyers who submits the highest bid is just the set of both buyers. Each buyer wins with probability $1/4$, because this is their probability of winning had they participated in a second price auction with two buyers and submitted a bid of $w_1$. As for the winning buyer's payment, in the first $R$ rounds of the phase, they will pay $w_1$ the second price, and in the second $R$ rounds of the phase they will pay $w_1 + 2(w_1-w_1) = w_1$, the second price plus the surcharge (which is zero in this phase). At the end of phase $1$, a buyer with value $w_1$ has cumulative utility $0$ for arm $b_{P}$, making the $w_1$-valued buyer ambivalent between arms $b_{P}$ and $b_{P-1}$. Meanwhile a buyer with value $w_i$ has positive cumulative utility $2R \cdot (w_i - w_1)$ for any $i > 1$ and strictly prefers $b_{P}$ to other arms. 

In phase $2$, the active arm range has increased to include $w_{P}$ and $w_{P-1}$. $w_P$ will now submit a bid of $w_2$ (in second price auction) when pulled, while $w_{P-1}$ will submit a bid of $w_1$. In phase 2, a buyer with value $w_2, w_3, w_4$ will continue to pull arm $b_P$ due to the previously accumulated utility, but a buyer with value $w_1$ will quickly learn to pull $b_{P-1}$ as its cumulative utility of arm $b_{P}$ becomes negative (because it sometimes gives the item to the buyer with a price above their value). In a similar manner, $w_2$ will become ambivalent between $b_{P}$ and $b_{P-1}$ towards the end of phase $2$ and will learn to switch to $b_{P-1}$ in phase $3$ (while $w_1$ will switch to $b_{P-2}$). 

By the start of phase $4$, the buyers will now pull distinct arms when they draw different values and the repeated auction enter the main phases. In the main phases the auction begins to extract revenue close to the full welfare by ensuring that the buyers get on average zero utility regardless of their value in each phase. For instance, in phase $\phaseVar = 5$, the arm $b_{P}$ has retired, and arms $b_1^{\phaseVar} = b_{P - 5 + 1}= b_{P - 4}$ to $b_4^{\phaseVar} = b_{P - 5 + 4} = b_{P-1}$ are the active arms representing different bids in the support of $\D$ (e.g. pulling $b_{P-1}$ submits a bid of $w_4$, which pulling $b_{P-4}$ submits a bid of $w_1$). The cumulative utility works out in a way where $w_{i}$'s preferred arm is $b_{P - 5 + i}$ for all $i \in [4]$. Let's say buyer one has value $w_3$, and buyer two has value $w_1$ in a particular round in the phase, then buyer one will pull arm $b_{P - 2}$ and buyer two will pull arm $b_{P-4}$. buyer one will get the item since they pulled a higher arm (and thus submitted a higher bid). If this particular round is among the first $R$ rounds, then buyer one will pay $w_1$, the second price. Otherwise, buyer one will pay $w_1 + 2(w_3 - w_1) = 2w_3 - w_1$, the second price plus the surcharge. Since there are roughly the same number of rounds where buyer one has value $w_3$ and buyer two has value $w_1$ during the first $R$ rounds vs in the second $R$ rounds of the phase, on average buyer one pays $1/2 (w_1 + 2w_3 - w_1) = w_3$ per round conditioned on buyer one having value $w_3$ and buyer two having value $w_1$. It is easy to generalize this and see that buyer one pays on average $w_3$ when its value is $w_3$, regardless of buyer two's values. 

\section{Omitted Proofs from Section~\ref{sec:FSE}}\label{app:fse}
Note that in this section we consider quantities that can be computed only based on $\mathcal{D}$ and number of buyers $n$ to be a constant (for instance $\min_{i} (w_i - w_{i-1})$, the expected utility from a second price auction with $n$ i. i.d buyers, given that our buyer's value is $v_j$). This is because $1/T$ can be made arbitrarily small compared to any of these quantities.  

\begin{proof}[Proof of Observation~\ref{obs:monotone}]
The allocation rule is obviously monotone, as higher-indexed arms submit a higher bid to the secondary auction (and then the secondary auction selects the highest bid). The payment rule is obviously monotone for the first $R$ rounds, since it is exactly a second-price auction.

For the second $R$ rounds, the initial payment from the second-price auction is again clearly monotone. The surcharge is also monotone, as the second-highest bid is fixed independently of the winner's arm, and $w_j$ is monotone in the arm selected. 

Finally, it is clear that the payment made remains monotone including the retired arms, as the payment of $2w_m$ exceeds the maximum possible value of $2(w_m-w_\ell) + w_\ell$ that an active arm might pay.
\end{proof}

We define $\Delta(\mathcal{D}) = \min_{i} (w_i - w_{i-1})$. For all remaining proofs in this section, we will assume $\alpha(\phaseVar) = o(1)$ for each phase $\phaseVar$, which implies that $\alpha(\phaseVar) \leq  \frac{\Delta(\mathcal{D})\cdot \XVCG(w_1)}{16nP} = \Theta(1)$. We will set $\alpha$ in our last proof. 

\begin{lemma} \label{app:lem:firstFSE}(precise statement of Lemma~\ref{lem:firstFSE}) Assume in $\phaseVar$ and each phase before $\phaseVar$, each buyer pulls a none intended arm only $\alpha(\phaseVar)$ fraction of rounds, then at the end of phase $\phaseVar$, the change in cumulative payoff of a buyer with value $v$ for each arm satisfies:
\begin{itemize}
\item For dormant arms $b$, $H_{2R\phaseVar}(v,b) - H_{2R(\phaseVar-1)}(v,b)= 0$.
\item For active arms: $H_{2R\phaseVar}(v,b_j^\phaseVar) - H_{2R(\phaseVar-1)}(v,b_j^\phaseVar)= 2R \cdot (v-w_j) \cdot \XVCG(w_j) \pm 2R \cdot 4n\alpha(\phaseVar)$.
\item For retired arms: $H_{2R\phaseVar}(v,b_j)  - H_{2R(\phaseVar-1)}(v,b_j)= 2R\cdot (v - 2w_m) \pm 2R \cdot 4n \alpha(\phaseVar)$.
\end{itemize} 
\end{lemma}

\begin{proof}
Note that the set up phase is designed in such a way that, when buyer $i$ pulls an active arm $b_j^\phaseVar$, assuming other buyers all pull their intended arms, buyer $i$'s expected utility in round $t$ of a set up phase is the same as that of round $t$ of the main phase. Therefore we do not need to distinguish these two kinds of phases when calculating historical utility in our proof. 

The first bullet follows trivially. To see this, observe that if an arm is dormant for phase $\phaseVar$, then it never awards the item or charges payments, so the cumulative utility doesn't change.

For active arms, observe that pulling arm $b_j^\phaseVar$ enters a bid of $w_j$ into a second price auction where all other bids are drawn independently from $\dist$ (except in a $n \cdot \alpha(\phaseVar)$ fraction of the rounds). This means that the buyer wins the item with probability $\XVCG(w_j)$ in each such round. 

In terms of payment, observe that the expected payment during the first $R$ rounds is exactly $\PVCG(w_j)$ per round. Moreover, the expected payment during the second $R$ rounds is exactly $\XVCG(w_j) \cdot (2w_j - w_\ell) = 2w_j\XVCG(w_j) - \PVCG(w_j)$ per round. We also observe that in a round $t$ when the other buyers are not pulling the desired arms, if $t$ is within the first $R$ rounds, the expected utility is in the range $[0, v]$, and if $t$ is within the second $R$ rounds, the expected utility is in the range $[\min(0, v - 2w_j), \max(0, v - w_j)]$. One can verify that the utility difference from expected is within $2(v + w_j) \leq 4$. Therefore, the added utility during phase $\phaseVar$ is:
\begin{align*}
 &2R\cdot v \cdot \XVCG(w_j) - R \cdot \PVCG(w_j) - R\cdot (2w_j \XVCG(w_j)) + R\cdot \PVCG(w_j)  \\
 & \hspace{8 cm}\pm  \cdot (2R) \cdot n\alpha(\phaseVar) \cdot 2(v+w_j)\\
 &= 2R \cdot (v-w_j) \cdot \XVCG(w_j)  \pm (2R) \cdot 4 n \alpha(\phaseVar).   
\end{align*}

For retired arms, observe that pulling the arm enters a bid of $w_m+1$ into a second-price auction, which surely wins in all rounds (except the $n \alpha(\phaseVar)$ fraction where other buyers may not pull the intended arm). The payment is $2w_m$. Therefore, the change in cumulative payoff is $2Rv - 4Rw_m \pm 2R \cdot  n\alpha(\phaseVar)(v - 2w_m) = 2Rv - 4Rw_m \pm 2R \cdot  4n\alpha(\phaseVar)$.
\end{proof}

\begin{corollary}\label{app:cor:firstFSE}(precise statement of Corollary~\ref{cor:firstFSE}) Assume in $\phaseVar$ and each phase before $\phaseVar$, each buyer pulls a none intended arm only $\alpha(\phaseVar)$ fraction of rounds, then at the end of phase $\phaseVar$, the cumulative payoffs for a buyer with value $v$ satisfy:
\begin{itemize}
\item If $b_j$ is dormant during phase $\phaseVar$ ($j \leq P-\phaseVar$): $H_{2R\phaseVar}(v,b)= 0$.
\item If $b_j$ is active during phase $\phaseVar$ ($j \in [P-\phaseVar +1,P-\phaseVar+m]$): 
$$H_{2R\phaseVar}(v,b_j) = 2R \cdot \left(\sum_{k=1}^{j +\phaseVar - P}(v-w_k) \cdot \XVCG(w_k) \right)\pm 8n \cdot \alpha(\phaseVar) \cdot T.$$
\item If $b_j$ is retired during $\phaseVar$ ($j \geq P-\phaseVar + m+1$): Let $a:=P+m-j$. Then:
$$H_{2R\phaseVar}(v,b) = 2R \cdot \left((\phaseVar-a)\cdot (v - 2w_m) + \sum_{k=1}^m  (v-w_k) \cdot \XVCG(w_k)  \right)\pm 8 n \cdot \alpha(\phaseVar) \cdot T .$$
\end{itemize}

\end{corollary}

\begin{proof}
The proof follows by just repeatedly applying Lemma~\ref{app:lem:firstFSE}. If $b_j$ is dormant, then $b_j$ was dormant in all previous rounds, so the claim trivially follows.

If $b_j$ is active, then it is intended for a buyer with value $w_{j +\phaseVar - P}$ during phase $\phaseVar$, a buyer with value $w_{j +\phaseVar - P - 1}$ in the previous phase, and in general intended for a buyer with value $w_{j + \phaseVar - \ell - P}$ in phase $\phaseVar-\ell$ for all $\ell \in [0, j + \phaseVar - P - 1]$ (where the intended buyer vary from $w_{j +\phaseVar - P}$ to $w_1$). Prior to being active, the arm was dormant. So the sum just sums up Lemma~\ref{app:cor:firstFSE} over all phases. 

If $b_j$ is retired, then it was retired for some number of previous rounds, and then prior to that it was active. $a = P +m-j$ is the last round where the arm was active. Once $a$ is computed, the sum follows by applying bullet three of Lemma~\ref{lem:firstFSE} to all retired rounds, and bullet two to the active rounds.
\end{proof}

\begin{lemma}\label{app:lem:favestart1} (precise statement of Lemma~\ref{lem:favestart} part 1) For all $\phaseVar > 1$, assume in each phase until $\phaseVar$, each buyer pulls a none intended arm only $\alpha(\phaseVar)$ fraction of rounds, then at the start of phase $\phaseVar$, a buyer with value $w_j$ \textbf{where} $j \leq \phaseVar$ has highest cumulative utility for arm $b_j^\phaseVar$, and also $b_j^{\phaseVar-1}$. Specifically, for all other arms $b_\ell$: 

$$H_{2R(\phaseVar-1)}(w_j,b_j^{\phaseVar-1}) \pm 16n \cdot \alpha(\phaseVar) \cdot T = H_{2R(\phaseVar-1)}(w_j,b_j^\phaseVar) > H_{2R(\phaseVar-1)}(w_j,b_\ell) + \Delta(\mathcal{D})\cdot \XVCG(w_{1})\cdot R. $$ 
\end{lemma}

\begin{proof}
This follows immediately from Corollary~\ref{app:cor:firstFSE}. Let $\err = 8n \cdot \alpha(\phaseVar) \cdot T$. We know that
\begin{align*}
    H_{2R(\phaseVar-1)}(w_j,b_j^\phaseVar) = H_{2R(\phaseVar-1)}(w_j,b_{P+j-\phaseVar}) = 2R \cdot \left(\sum_{k=1}^{j-1}(w_j-w_k) \cdot \XVCG(w_k) \right)\pm \err.
\end{align*}
First we will show our claim for all active arm $b_\ell^\phaseVar$ where $\ell \neq j$. 
\begin{align*}
    H_{2R(\phaseVar-1)}(w_j,b_\ell^\phaseVar) = H_{2R(\phaseVar-1)}(w_j,b_{P+\ell-\phaseVar}) = 2R \cdot \left(\sum_{k=1}^{\ell-1}(w_j-w_k) \cdot \XVCG(w_k) \right)\pm \err.
\end{align*}
Thus 
\begin{align*}
    &H_{2R(\phaseVar-1)}(w_j,b_j^\phaseVar) - H_{2R(\phaseVar-1)}(w_j,b_\ell^\phaseVar) \\
    &=  2R \cdot \left(\sum_{k=1}^{j-1}(w_j-w_k) \cdot \XVCG(w_k) \right) - 2R \cdot \left(\sum_{k=1}^{\ell-1}(w_j-w_k) \cdot \XVCG(w_k) \right) \pm \err\\
    &=2R \cdot \left(\ind [\ell \leq j] \cdot  \sum_{k=\ell}^{j-1}(w_j-w_k) \cdot \XVCG(w_k) + \ind [\ell > j] \sum_{k=j}^{\ell-1}(w_k - w_j) \XVCG(w_k) \right) \pm \err.\\
\end{align*}
we conclude that when $l = j+1$, $H_{2R(\phaseVar-1)}(w_j,b_j^\phaseVar) = H_{2R(\phaseVar-1)}(w_j,b_\ell^\phaseVar) \pm err$. Meanwhile, while $l \neq j+1$, 
\begin{align*}
    &H_{2R(\phaseVar-1)}(w_j,b_j^\phaseVar) - H_{2R(\phaseVar-1)}(w_j,b_\ell^\phaseVar)\\
    &\geq 2R \cdot \Delta(\mathcal{D}) \cdot \XVCG(w_1) \pm err \\
    &\geq 2R \cdot \Delta(\mathcal{D}) \cdot \XVCG(w_1) - 8n \cdot T \cdot \frac{\Delta(\mathcal{D})\cdot \XVCG(w_1)}{16nP} \geq R \cdot \Delta(\mathcal{D}) \cdot \XVCG(w_1).
\end{align*}
Now we show that all dormant or retired arms are strictly worse by $R \cdot \Delta(D) \cdot \XVCG(w_1)$ compared to one of the active arms. Firstly, for all retired arm $b_\ell$ in phase $\phaseVar$, 
\begin{align*}
    H_{2R(\phaseVar-1)}(w_j,b_\ell) &= 2R \cdot \left((\phaseVar-a)\cdot (w_j - 2w_m) + \sum_{k=1}^m  (w_j-w_k) \cdot \XVCG(w_k)  \right)\pm \err \\
    &\leq -2R \cdot  \Delta(\mathcal{D}) +  H_{2R(\phaseVar-1)}(w_j,b_m^{\phaseVar-1}) \pm 2\cdot \err \\
    &\leq H_{2R(\phaseVar-1)}(w_j,b_m^{\phaseVar-1}) - R \cdot \Delta(\mathcal{D}) \cdot \XVCG(w_1).
\end{align*}
Secondly, from our calculation of active arms, it is easy to see that $H_{2R(\phaseVar-1)}(v,b_1^{\phaseVar-1}) \geq R \cdot \Delta(\mathcal{D}) \cdot \XVCG(w_1) $, while when an arm $b_\ell$ is dormant, $H_{2R(\phaseVar-1)}(v,b_\ell) = 0$. This completes our proof. 
\end{proof}

\begin{lemma}(precise statement of Lemma~\ref{lem:favestart} part 2) \label{app:lem:favestart2}
For all $\phaseVar > 1$, assume in each phase until $\phaseVar$, each buyer pulls a none intended arm only $\alpha(\phaseVar)$ fraction of rounds, then at the start of phase $\phaseVar$, a buyer with value $w_j$ \textbf{where} $j > \phaseVar$ has highest cumulative utility for arm $b_p = b_\phaseVar^\phaseVar$. Specifically, for all $b_\ell \neq b_{\phaseVar}^{\phaseVar}$, 
$$ H_{2R(\phaseVar-1)}(w_j,b_\phaseVar^\phaseVar) > H_{2R(\phaseVar-1)}(w_j,b_\ell) + \Delta(\mathcal{D})\cdot \XVCG(w_{1})\cdot R .$$
\end{lemma}

\begin{proof}
Let $\err = 8n \cdot \alpha(\phaseVar) \cdot T$. Since $j > \phaseVar$, we know that $\phaseVar$ is a set up phase and $b_P = b_\phaseVar^\phaseVar$ is active. From the proof for Lemma~\ref{app:lem:favestart1}, we know that all dormant and retired arms are strictly worse than $b_\phaseVar^\phaseVar$ by $R \cdot \Delta(\mathcal{D})\cdot \XVCG(w_{1})$. For any active arm $b_\ell^\phaseVar \neq b_\phaseVar^\phaseVar$, $\ell < \phaseVar$. Therefore 
\begin{align*}
    &H_{2R(\phaseVar-1)}(w_j,b_j^\phaseVar) - H_{2R(\phaseVar-1)}(w_j,b_\ell^\phaseVar) \\
    & = 2R \cdot \left(\sum_{k=1}^{\phaseVar-1}(w_j-w_k) \cdot \XVCG(w_k) \right) - 2R \cdot \left(\sum_{k=1}^{\ell-1}(w_j-w_k) \cdot \XVCG(w_k) \right) \pm \err\\
    &= 2R \cdot \left(\sum_{k=\ell}^{\phaseVar-1}(w_j-w_k) \cdot \XVCG(w_k) \right) \pm \err \geq 2R \cdot \Delta(\mathcal{D})\cdot \XVCG(w_{1}) \pm err\\
    &\geq R \cdot \Delta(\mathcal{D}) \cdot \XVCG(w_1).
\end{align*}
\end{proof}
\begin{corollary}\label{app:cor:favefirst1} For all $\phaseVar \geq 1$, assuming that all other buyers pull their intended arm except for $\alpha(\phaseVar-1) > \gamma$ fraction of rounds before phase $\phaseVar$, a mean-based buyer with value $w_j$, \textbf{where $j \leq \phaseVar$}, pulls arm $b_j^\phaseVar$ for the first $R$ rounds, except for at most $C_1 \cdot \alpha(\phaseVar - 1) \cdot R$ rounds for some constant $C_1$ that is not dependent on $\alpha$ or $\phaseVar$.
\end{corollary}
\begin{proof}
Firstly, observe that by lemma~\ref{app:lem:favestart1}, by the end of phase $\phaseVar - 1$, for any arm $b_\ell \neq b_j^\phaseVar, b_j^{\phaseVar - 1}$, a mean-based buyer with value $w_j$ will have 
\begin{align*}
H_{2R(\phaseVar-1)}(w_j,b_j^{\phaseVar})  &> H_{2R(\phaseVar-1)}(w_j,b_\ell) + \Delta(\mathcal{D})\cdot \XVCG(w_{1})\cdot\frac{T}{P} \\
&= H_{2R(\phaseVar-1)}(w_j,b_\ell) + \Theta(T) >  H_{2R(\phaseVar-1)}(w_j,b_\ell) + \gamma T.
\end{align*}

Therefore buyer with value $w_j$ will pull the arm $b_\ell$ with probability $< \gamma$. We conclude that buyer with value $w_j$ may only pull arms $b_j^{\phaseVar - 1} = b_{j+1}^{\phaseVar}$ and $b_j^{\phaseVar}$ with constant probability in all $R$ rounds. 

Now we claim that buyer with value $w_m$ will be pulling arm $b_{m+1}^{\phaseVar}$ with probability at most $\gamma$ after the beginning $x = \frac{20Tn}{v_m} \alpha(\phaseVar - 1)$ rounds. This is because in each round where $w_m$ pulls arm $b_{m+1}^{\phaseVar}$, the buyer's cumulative utility of arm $b_{m+1}^{\phaseVar}$ decreases by at least $w_m$. Meanwhile, if $w_m$ pulls arm $b_m^{\phaseVar}$, its utility is always positive. Therefore after $x$ rounds, 
$$H_{2R(\phaseVar-1) + x}(w_m,b_m^{\phaseVar}) - H_{2R(\phaseVar-1) + x}(w_m,b_{m+1}^{\phaseVar}) \geq -16n \cdot \alpha(\phaseVar-1)\cdot T + w_m \cdot \frac{20 Tn}{w_m} \alpha(\phaseVar-1) \geq \gamma \cdot T.$$

Next we prove that given all buyers with value $w_{j+1}$ will pull arm $b_{\ell} \neq b_{j+1}^{\phaseVar}$ with probability $< \gamma$ after round $t$, all buyers with value $w_{j}$ will pull arm $b_{\ell} \neq b_{j}^{\phaseVar}$ with probability $< \gamma$ after round $t + x$ where $x = O(\alpha(\phaseVar - 1)T)$. 

Consider a buyer with value $w_j$ participating in a second price auction. Note that the only difference between bidding $w_j$ and $w_{j+1}$ is: 1) bidding $w_{j+1}$ can win the item with higher probability when second price is $w_j$, but the utility gain for both strategies is $0$; and 2) bidding $w_{j+1}$ can win the item with positive probability when second price is $w_{j+1}$, in this case the buyer will incur negative utility. Since all buyers bid $b_{j+1}^\phaseVar$ when their value is $w_{j+1}$ with high probability, second price is $w_{j+1}$ with a constant probability $\geq (1 - n\gamma) (q_{j+1})^{(n-1)}$ in each round (we view everything related to $\mathcal{D}$ as a constant). We conclude that in any round $t' > t$, for buyer $i$, expected utility difference between arm $b_j^\phaseVar$ and $b_{j+1}^\phaseVar$ in each round $E[r_{i,b_j^\phaseVar, t'}(w_j) - r_{i,b_{j+1}^\phaseVar, t'}(w_j)]$ is lower bounded by some constant  $\Delta(r   (w_j))$ not dependent on $t$.  

Now notice that in the first $R$ rounds of phase $\phaseVar$ essentially runs a second price auction in each round, where pulling arm $b_{j}^\phaseVar$ will always give weakly better utility compared to pulling any other arm. Also, by our assumption, all other buyers will bid their intended arm $w.p. > 1 - mn\gamma$ after round $t$. Let $x = \frac{20Tn}{\Delta(r(w_j))} \alpha(\phaseVar - 1)$, then 
\begin{align*}
    &H_{2R(\phaseVar-1) + x}(w_j,b_j^{\phaseVar}) - H_{2R(\phaseVar-1) + x}(w_j,b_{j+1}^{\phaseVar})\\ &\geq -16n \cdot \alpha(\phaseVar-1)\cdot T +  \Delta(r(w_j)) x  \\
    &\geq -16n \cdot \alpha(\phaseVar-1)\cdot T  + \Delta(r(w_j)) \frac{20 Tn}{\Delta(r(w_j))} \alpha(\phaseVar-1) \geq \gamma \cdot T.
\end{align*}
Therefore after round $t + x$, buyer with value $w_{j}$ will pull all arms $b_{\ell}$, including $b_{j}^\phaseVar$, with probability $< \gamma$. 
Finally we conclude that after $C' \alpha(\phaseVar - 1) R = m \cdot \frac{20Tn}{\Delta(r(w_j))} \alpha(\phaseVar - 1)$ rounds, all buyers will pull any unintended arm with probability $< \gamma$. Hence during the first $R$ rounds in phase $\phaseVar$, a buyer can pull an unintended arm in only $C_1 \alpha(\phaseVar - 1) R = C' \alpha(\phaseVar - 1) R + \gamma R$ rounds. 
\end{proof}

\begin{lemma}\label{app:cor:favefirst2} For all $\phaseVar > 1$, assuming that all other buyers pull their intended arm except for $\alpha(\phaseVar-1) > \gamma$ fraction of rounds before phase $\phaseVar$, a mean-based buyer with value $w_j$, \textbf{where $j > \phaseVar$}, pulls arm $b_p = b_\phaseVar^\phaseVar$ for the first $R$ rounds, except for at most $\gamma R$ rounds.
\end{lemma}
\begin{proof}
By lemma~\ref{app:lem:favestart2}, by the end of phase $\phaseVar - 1$, for any arm $b_\ell \neq b_\phaseVar^\phaseVar$, a mean-based buyer with value $w_j$ will have 
\begin{align*}
H_{2R(\phaseVar-1)}(w_j,b_\phaseVar^{\phaseVar})  &> H_{2R(\phaseVar-1)}(w_j,b_\ell) + \Delta(\mathcal{D})\cdot \XVCG(w_{1})\cdot\frac{T}{P}\\
&= H_{2R(\phaseVar-1)}(w_j,b_\ell) + \Theta(T) >  H_{2R(\phaseVar-1)}(w_j,b_\ell) + \gamma T.
\end{align*}

Therefore at the start of phase $\phaseVar$, buyer with value $w_j$ will pull the arm $b_\ell$ with probability $< \gamma$. During the first $R$ rounds of phase $\phaseVar$, pulling arm $b_\phaseVar^\phaseVar$ remains a weakly dominating strategy in each round, thus $w_j$ will continue to pull arm $b_\ell$ with probability $< \gamma$. 
\end{proof}

\begin{corollary}\label{app:lem:setupfirst}(precise statement of Corollary~\ref{lem:setupfirst}) For all $\phaseVar$, assuming that all other buyers pull their intended arm except for $\alpha(\phaseVar-1) > \gamma$ fraction of rounds before phase $\phaseVar$, a mean-based buyer with value $w_j$ pulls arm $b_j^\phaseVar$ (if it exists, otherwise they pull arm $b_P = b_\phaseVar^\phaseVar$) for the first $R$ rounds, except for at most $C_1 \cdot \alpha(\phaseVar - 1) \cdot R$ rounds for some constant $C_1$ that is not dependent on $\alpha$ or $\phaseVar$.
\end{corollary}
\begin{proof}
By Lemma~\ref{app:cor:favefirst1} and~\ref{app:cor:favefirst2}. 
\end{proof}

\begin{lemma} \label{app:lem:setupend} (precise statement of Lemma~\ref{lem:setupend}) For all phase $\phaseVar > 1$, assuming that all other buyers pull their intended arm except for $\alpha(\phaseVar-1)$ fraction of rounds before phase $\phaseVar$, a mean-based buyer with value $w_j$ pulls arm $b_j^\phaseVar$ (if it exists, otherwise they pull arm $b_P$) for the last $R$ rounds, except for at most $C \cdot \alpha(\phaseVar - 1) \cdot R$ rounds for some constant $C_2$ that is not dependent on $\alpha$ or $\phaseVar$.
\end{lemma}
\begin{proof}
To prove this, we simply observe that Corollary~\ref{app:lem:setupfirst} establishes that
all buyers pull an intended arm with probability $> 1 - mn\gamma$ after first $adj(\phaseVar) = C' \alpha(\phaseVar - 1) R$ rounds in the phase. Let $b_a^\phaseVar$ be the intended arm for buyer with value $w_j$, then $a = \min(\phaseVar, j)$. It is easy to see that because the errors terms are fully adjusted by the first $adj(\phaseVar)$ rounds of auctions, 
\begin{align*}
    &H_{2R(\phaseVar-1) + adj(\phaseVar)}(w_j,b_a^{\phaseVar}) - H_{2R(\phaseVar-1) + adj(\phaseVar)}(w_j,b_{\ell}^{\phaseVar}) \\
    &\geq \left(\ind [\ell < a] \cdot  \sum_{k=\ell}^{j-1}(w_j-w_k) \cdot \XVCG(w_k) + \ind [\ell > a] \sum_{k=a}^{\ell-1}(w_k - w_j) \XVCG(w_k) \right).
\end{align*}
We define $$\Delta(r(w_j), r(w_\ell)) = w_j \cdot \XVCG(w_j) - \PVCG(w_j) - w_j \XVCG(w_\ell) + \PVCG(w_\ell).$$ Assume this phenomenon that that
all buyers pull an intended arm with probability $> 1 - mn\gamma$ last until round $R + x$, then for any $x$ where $x < \left(1 - \frac{8Pn\gamma}{\Delta(r(w_j), r(w_\ell))}\right)R - adj(\phaseVar)$ and active arm $b_{\ell}^{\phaseVar}$,  
\begin{align*}
   &H_{2R(\phaseVar-1) + R+x}(w_a,b_a^{\phaseVar}) - H_{2R(\phaseVar-1) + R+x}(w_a,b_{\ell}^{\phaseVar}) \\
   &\geq 2R \cdot \left(\ind [\ell \leq j] \cdot  \sum_{k=\ell}^{a-1}(w_a-w_k) \cdot \XVCG(w_k) + \ind [\ell > j] \sum_{k=j}^{\ell-1}(w_k - w_a) \XVCG(w_k) \right)\\
   & \quad \quad + (1 - mn\gamma) \cdot (R - adj(\phaseVar)) \left(w_a \cdot \XVCG(w_a) - \PVCG(w_a) - w_a \XVCG(w_\ell) + \PVCG(w_\ell)\right) \\
   & \quad \quad \text{(utility difference in first R rounds when all buyers bid as intended)}\\
&  \quad \quad + (1 - mn\gamma)\cdot x \cdot (w_a \cdot \XVCG(w_a) - 2w_a \XVCG(w_a) + \PVCG(w_a))  \\
&  \quad \quad - (1 - mn\gamma)\cdot x \cdot (w_a \cdot \XVCG(w_\ell) - 2w_\ell \XVCG(w_\ell) + \PVCG(w_\ell))\\
&  \quad \quad \text{(utility difference in last $x$ rounds when all buyers bid as intended)} \\
&  \quad \quad - mn\gamma x\cdot 4 \text{ (rounds when buyers do not bid as intended)}\\
&\geq  2R \cdot \left(\ind [\ell < j] \cdot  \sum_{k=\ell}^{a-1}(w_a-w_k) \cdot \XVCG(w_k) + \ind [\ell > j] \sum_{k=j}^{\ell-1}(w_k - w_a) \XVCG(w_k) \right) \\
&  \quad \quad - (1 - mn\gamma)\cdot x \cdot 2(w_a - w_\ell)\cdot \XVCG(w_\ell) \\
&  \quad \quad + (R - adj(\phaseVar)-x) \cdot \Delta(r(w_a), r(w_\ell)) - mn\gamma x\cdot 4\\
& \geq  2R \cdot \left(\ind [\ell < j] \cdot  \sum_{k=\ell+1}^{a-1}(w_a-w_k) \cdot \XVCG(w_k) + \ind [\ell > j] \sum_{k=j}^{\ell}(w_k - w_a) \XVCG(w_k) \right) \\
&  \quad \quad \text{(Note that this term is non negative)}\\
&  \quad \quad + (R - adj(\phaseVar)-x) \cdot \Delta(r(w_a), r(w_\ell)) - mn\gamma x\cdot 4\\
& \geq (R - adj(\phaseVar)-x) \cdot \Delta(r(w_a), r(w_\ell)) - 4mn x\gamma \geq \gamma T. 
\end{align*} 
This is consistent with $b_a^\phaseVar$ being strictly the favourite arm of buyer with value $w_j$. We conclude that a buyer with value $w_j$ will pull their intended arm aside from $adj(\phaseVar) + \frac{8Pn\gamma}{\Delta(r(w_j), r(w_\ell))}R = C_2 \cdot \alpha(\phaseVar - 1)$ rounds for some constant $C_2$.  

\end{proof}

\begin{proof}[Proof of Proposition~\ref{prop:behave}]
Let us define $\alpha(1) = 2\gamma$ and $\alpha(\phaseVar) = (C_1 + C_2) \cdot \alpha(\phaseVar - 1)$ for $\phaseVar > 1$. It is easy to see that $\alpha(\phaseVar)$ is increasing. Firstly, in phase $1$ there is only one arm that is active and everyone will try to pull that arm aside from $2\gamma R$ rounds. Next, assume for each phase $i \leq \phaseVar-1$ each buyer pull their intended arm aside from $\alpha(i) R$ rounds. Then since $\alpha$ is increasing, on average each buyer pull their intended arm in each round aside from $\alpha(\phaseVar-1) \phaseVar R$ rounds by the end of phase $\phaseVar-1$. By corollary~\ref{app:lem:setupfirst} and lemma~\ref{app:lem:setupend}, each buyer will pull their intended arm aside from at most $C_1 \alpha(\phaseVar-1) R + C_2 \alpha(\phaseVar-1) R = \alpha(\phaseVar) R$ rounds. Finally, we verify that $\alpha(\phaseVar) = o(1)$ for all $\phaseVar$. Note that there are in total only constant ($P$) number of phases, therefore $\alpha(\phaseVar) = (C_1 + C_2)^{\phaseVar-1} \alpha(1) \leq (C_1 + C_2)^P \cdot 2 \gamma = o(1)$. This completes our proof. 
\end{proof}
\section{Omitted Proofs from Section~\ref{sec:UB}}\label{app:UB}

\begin{theorem}(restatement of Theorem~\ref{thm:counterexample}) There exist $(x,y,u)$ that satisfy the $n$-buyer BMSW Constraints but are not $n$-buyer auction feasible, and $(x,u)$ that satisfy the Reduced $n$-buyer BMSW Constraints but are not $n$-buyer auction feasible.
\end{theorem}

\begin{proof}
Let distribution $\mathcal{D}$ be such that the support $w = [1/M, 4/M, 5/M, 10/M]$ for some large number $M$ and $q = [5\delta, \delta, \delta, 1 -7\delta]$ for some $\delta < 1/7$. Let $x^*$ be the border extreme point $x^* = W q$, where $$W =
\begin{bmatrix}
	0.5 & 0 & 0 & 0\\
	1 & 0.5 & 0 & 0\\
	1 & 1 & 0.5 & 0\\
	1 & 1 & 1 & 0.5
\end{bmatrix}.$$ 


Let $\ust(\xst)$ be such that for all $i$, $$\ust_i = \min_{j \le i} \spc (w_i - w_j)x_j.$$ By construction $(x^*, u^*(\xst))$ satisfy the Reduced $n$-buyer BMSW Constraints. We can explicitly calculate $\xst$ to be $\xst = [1-7\delta/2, 13\delta/2, 11\delta/2, 5\delta/2]$. One can verify that  $\xst$ satisfy  
\begin{align*}
	10 \delta \cdot 1/M= 5\delta/2\ \cdot 4/M = \xst_1 (w_3 - w_1) &> \xst_2 (w_3 - w_2) = 11\delta/2 \cdot 1/M = 5.5 \delta \cdot 1/M  , \\
	33 \delta \cdot 1/M= 11\delta/2 \cdot 6/M = \xst_2 (w_4 - w_2) &> \xst_1 (w_4 - w_1) = 5\delta/2 \cdot 9/M = 22.5 \delta \cdot 1/M,\\
	33 \delta \cdot 1/M = 11\delta/2 \cdot 6/M = \xst_2 (w_4 - w_2) &> \xst_3 (w_4 - w_3) = 13\delta/2 \cdot 5/M= 32.5 \delta \cdot 1/M.
\end{align*}
Thus by Claim~\ref{claimLPSolutionInvalid}, $(\xst, \ust(\xst))$ is not $n$-buyer auction feasible. Note that $(\xst, \xst, \ust(\xst))$ also serves as an example that is a feasible solution to the n-buyer BMSW Constraints but not $n$-buyer auction feasible. 
\end{proof}

\begin{claim} \label{claimLPSolutionInvalid}
Given distribution $\mathcal{D}$, if a border extreme point $x^*$ satisfy \begin{align*}
	\xst_1 (w_3 - w_1) &> \xst_2 (w_3 - w_2),\\
	\xst_2 (w_4 - w_2) &> \xst_1 (w_4 - w_1),\\
	\xst_2 (w_4 - w_2) &> \xst_3 (w_4 - w_3)
\end{align*}
and for all $i$, $\ust_i = \min_{j \le i} \spc (w_i - w_j) x_j$, then $(x^*, u^*)$ is not $n$-buyer auction feasible for distribution $\mathcal{D}$. 
\end{claim}

\begin{proof}
Assume $(\xst, \ust)$ is $n$-buyer auction feasible, and let $A$ the be the repeated auction where when the buyers are running clever $\gamma$-mean based algorithms, $$X_{ij}^A(\dist, S_1,\ldots, S_n)= \xst_j \pm O(\gamma); U_{ij}^A(\dist, S_1,\ldots, S_n) =\ust_j \pm O(\gamma).$$ Therefore when the buyers are running clever $0$-mean based algorithms, the average allocation probability and utility for buyer with value $w_j$ in $A$ is exactly $\xst_j$ and $\ust_j$. 

Because $\xst$ is an extreme point of the border polytope, $\xst$ cannot be written as a convex combination of $\geq 2$ distinct points that also satisfy the Border constraints. Since $\xst$ is a convex combination of the allocation probabilities in each fixed round, all the allocations during each round must be $\xst$ itself. This means that any auction which matches $(\xst,\ust)$ must have buyers with value $w_j$ receive the item with probability $x_j$ in every single round (except for $o(T)$). 

Suppose that at the end of auction $A$, the buyer with value $w_3$ is pulling an arm $a$ with label $\ge w_2$. 
Then because $w_3$ buyer's favorite arm is $a$ at the end, the $w_3$ buyer gets utility $\ust_3$ on average (across rounds of auctions) by pulling arm $a$ always. Furthermore, arm $a$'s average allocation is at least $\xst_2$
(because buyer with value $w_2$ gets the item with probability $\xst_2$ every round by pulling arms $\leq w_2$, and the allocation probability is monotone in value in each round). Thus, buyer with value $w_4$ can get utility  $\xst_2(w_4 - w_3) + \ust_3$ by pulling arm $a$ the whole way through. Then we reach a contradiction:
\begin{align*}
	\ust_4 &\geq \xst_2 (w_4 - w_3) + \ust_3 \\
	&\geq \xst_2 (w_4 - w_3) + \xst_1(w_3 - w_1)\\
	&> \xst_2 (w_4 - w_3) + \xst_2(w_3 - w_2)\\
	&\geq  \xst_2 (w_4 - w_2) = \ust_4.
\end{align*}

Suppose instead that buyer with value $w_3$ at the end has strictly higher utility for 
an arm $a$ labeled $w_1 \leq a < w_2$ than for any arm outside of this region. 
Then the $w_3$ buyer must be pulling arm $a$ for a non-empty period at the 
end of the auction. 
Thus, in a non-empty set of rounds, arm $a$ grants allocation $\xst_3$.
During those rounds, arm $w_2$ must also grant allocation $\geq \xst_3$ 
(by monotonicity).
Then because arm $w_2$ must grant at least $\xst_2$ every round, it grants 
$> \xst_2$ allocation on average throughout the auction, 
while charging at most $w_2$. This means that buyer with value $w_4$ can get average utility strictly greater than 
$\xst_2 (w_4 - w_2)$ by pulling arm $w_2$ every round,
contradicting $\ust_4 = \xst_2 ( w_4 - w_2)$. 

Finally, suppose the $w_3$ buyer has the highest utility at the end 
for an arm $a$ labeled $< w_1$ (not necessarily strictly higher than any arm outside of this region). 
Then arm $a$ must give the item with positive probability on average. Since $a$ is labeled $< w_1$, arm $a$ always charges less than $w_1$. This means that $\ust_1 > 0$, contradicting $\ust_1 = min_{j \leq 1} (w_1 - w_j) x_j = 0$. 
\end{proof}



\section{Omitted Proofs from Section~\ref{sec:UPYB}}\label{app:UPYB}
Note that the notion of optimal is with regard to $0$-regret buyers, therefore we will assume the buyers are $0$-regret throughout the proofs in this section. 

\begin{theorem} (restatement of Theorem~\ref{thm:uniform2}) The optimal pay-your-bid uniform auction with declining reserve can be found by a linear program of size $\poly(m)$. However, there exist $2$-buyer instances where the optimal $2$-buyer feasible auction strictly outperforms the best pay-your-bid uniform auction with declining reserve.
\end{theorem}
\begin{proof}
By Claim~\ref{app:claimUniformLP} and Claim~\ref{app:claimUniformSubopt}.
\end{proof}
\begin{claim}\label{app:claimUniformLP}
The optimal pay-your-bid uniform auction with declining reserve can be found by a linear program of size $\poly(m)$. 
\end{claim}
\begin{proof}

In uniform auction with declining reserve, every buyer above the reserve has an equal chance of winning the item (and paying their bid) in each round. Let us define variables $l_1, \dots, l_m$, where $l_i$ denotes the fraction of rounds in which bidding $w_i$ would win the item with positive probability (i.e. the fraction of rounds in which $w_i$ is above the reserve). Let $l_{0} = 0$. For monotonicity, $l_i \leq l_{i + 1}$ $\forall i$, for the $l_1, ... l_m$ representing disjoint segment of auctions, $\sum_i (l_i - l_{i - 1}) \leq 1$. 

Assume there are in total $T$ rounds of auctions. Then for the first $(1 - l_i)T$ rounds, bidding $w_i$ would win the item with probability 0. Let $H_j^{n - 1}$ be the random variable representing the number of buyers with type $\geq w_j$ among the $n - 1$ other buyers in the auction \footnote{$H_j^{n - 1}$ can be thought of as a binomially distributed random variable with parameters $n - 1$, $p = 1 - F_{\dist}(w_j)$.}. Then during the interval $((1 - l_i)T, (1 - l_{i - 1})T)$, a buyer with type $w_j \geq w_i$ gets the item each round with probability $\E\left[\frac{1}{1 + H_i^{n - 1}}\right]$. Therefore for any clever buyer strategies $S_1, ... S_n$ and for buyer $i$, 
\begin{align*}
	X_{ij}^A(\dist, S_1, \ldots, S_n) &= \sum_{i = 1}^{j} (l_i - l_{i - 1})\E\left[\frac{1}{1 + H_i^{n - 1}}\right].
\end{align*}
By setting $\lambda_i = l_i - l_{i - 1}$, we can rewrite the relationship between $X$ and $\lambda$ as 
\begin{align*}
	X_{ij}^A(\dist, S_1, \ldots, S_n) = \sum_{i = 1}^{j} \lambda_{i} \E\left[\frac{1}{1 + H_i^{n - 1}}\right], \text{ subject to } \lambda_i \geq 0, \sum_i \lambda_i \leq 1.
\end{align*}
Observe that $\E\left[\frac{1}{1 + H_i^{n - 1}}\right]$ actually can be calculated solely based on $\mathcal{D}$. Specifically, let $Q_i= \sum_{k=i}^{m} q_i$ be the probability that a buyer has value $\geq w_i$, then  
\begin{align*}
\E\left[\frac{1}{1 + H_i^{n - 1}}\right] = \sum_{j=0}^{n-1} \frac{1}{1 + j} \cdot \binom{n-1}{j} \cdot  Q_i^j \cdot  (1-Q_i)^{n-1-j}. 
\end{align*}

Denote the value of $\E\left[\frac{1}{1 + H_k^{n - 1}}\right]$ as $E_{k}$. We observe that 
$$ X_{ij}^A(\dist, S_1, \ldots, S_n) - X_{i(j-1)}^A(\dist, S_1, \ldots, S_n) = \lambda_j E_j,$$
let variable $x_j$ represent the value of $X_{ij}^A(\dist, S_1, \ldots, S_n)$ where $S_1, ... S_n$ have $0$-regret for any buyer $i$, then $\lambda_j = \frac{x_j - x_{j-1}}{E_j}$ (we define $x_0 = 0$ here). The constraint $\sum_{k=1}^m \lambda_k \leq 1$ can now be written as $\sum_{k=1}^m \frac{x_j - x_{j-1}}{E_j} \leq 1$, which is equivalent to $$\sum_{j=1}^m x_j (\frac{1}{E_j} - \frac{1}{E_{j+1}}) \leq 1,$$ where we define $1/E_{j+1} = 0$. We conclude that solutions $\lambda$ from the following linear program correspond to the optimal pay-your-bid uniform auction with declining reserve. 
 
\newcommand{\lpUnifReserveReduced}{\ReducedUniformLP}
\begin{align*}
	\mathbf{maximize }_{\vec{x}, \vec{u}} \spc \spc \spc & \sum_{i = 1}^m q_i (w_i x_i - u_i) \tag{\textbf{\lpUnifReserveReduced}}\\
	\mathbf{subject \spc to \spc} \spc \spc \spc &
	\sum_{j=1}^m x_j (\frac{1}{E_j} - \frac{1}{E_{j+1}}) \leq 1\\
	& u_i \geq (w_i - w_j) \cdot x_j, \forall i, j \in [m] : i > j\\
	& \vec{x} \text{ is monotone}\\
	& \vec{x}, \vec{u} \geq 0.
\end{align*} 
\end{proof}

\begin{claim}\label{app:claimUniformSubopt}
There exist $2$-buyer instances where the optimal $2$-buyer feasible auction strictly outperforms the best pay-your-bid uniform auction with declining reserve.
\end{claim}
\begin{proof}
Consider the following example with $2$ buyers with value distribution $\mathcal{D}$ where
\begin{align*}
    w = \left[\frac{1}{4}, \frac{2}{4}, \frac{3}{4}, 1\right], q = \left[\frac{1}{4}, \frac{1}{4}, \frac{1}{4}, \frac{1}{4}\right].
\end{align*}
One can verify that the optimal solution to \ReducedUniformLP{} for the example is $x_1 = x_2 = 0, x_3 = x_4 = \frac{3}{4}$, where the attained revenue is $9/16$. This solution corresponds to running uniform auction with fixed reserve at $3/4$ in each round. However, simply running second price auction with fixed reserve at $3/4$ in each rounds results in an increase in expected revenue. 
To see this, let us compare the revenue from both auction formats for each fixed pair of buyer values. We observe that in no case could second price generate worse revenue than uniform price auction, given that they have the same reserve. Moreover, when both buyers have value $1$, second price auction attains revenue $1$ while uniform auction only attains revenue $\frac{3}{4}$. Hence we conclude that pay-your-bid uniform auction with declining reserve is not the optimal auction format. 
\end{proof}

\section{Omitted Proofs from Section~\ref{sec:nonconvex}}\label{app:nonconvex}

\begin{theorem}[restatement of Theorem~\ref{thm:nonconvex}] Let $P$ denote the set of all tuples $(x,y,u)$ that are $n$-buyer feasible auctions \emph{where the bid space is equal to the support of $\dist$}. Then $P$ is not necessarily convex, even when $n=2$. 
\end{theorem}
\begin{proof}



\newcommand{\xone}{x^a}
\newcommand{\xtwo}{x^b}
\newcommand{\yone}{y^a}
\newcommand{\ytwo}{y^b}
\newcommand{\aone}{A^a}
\newcommand{\atwo}{A^b}
Consider $n = 2$ (there are 2 buyers) and distribution $\mathcal{D}$ satisfying $w = [1, 3, 4, 7, 30]$ and $q = [1/5, 1/5, 1/5, 1/5, 1/5]$. Let 
\begin{align*}
    \xone = \begin{bmatrix}
        9/10\\
        7/10\\
        3/10\\
        3/10\\
        3/10
    \end{bmatrix}, 
    \yone = \begin{bmatrix}
        9/10\\
        9/10\\
        9/10\\
        7/10\\
        3/10\\
    \end{bmatrix}; \quad \quad \quad \quad 
    \xtwo = \begin{bmatrix}
        9/10\\
        7/10\\
        3/10\\
        3/10\\
        3/10
    \end{bmatrix}, 
    \ytwo = \begin{bmatrix}
        9/10\\
        9/10\\
        7/10\\
        3/10\\
        3/10\\
    \end{bmatrix}.
\end{align*}
We show that $(\xone, \yone)$, $(\xtwo, \ytwo)$ are both $n$-buyer auction feasible. However, $(\frac{\xone + \xtwo}{2}, \frac{\yone + \ytwo}{2})$ is not $n$-buyer auction feasible. Let $\aone$ be the auction where a $\gamma$-mean buyer gets the item with probability $\yone_i + O(\gamma)$ by bidding $w_i$ in each round. Firstly, consider auction $A^a$.
\begin{align*}
    &(w_5 - w_4) \cdot \yone_4 = 23 \cdot 9/10 = 20.7\\
   &(w_5 - w_3) \cdot \yone_3 = 26 \cdot \frac{9}{10} = 23.4\\
    &(w_5 - w_2) \cdot \yone_2 = 27 \cdot \frac{7}{10} = 18.9\\
    &(w_5 - w_1) \cdot \yone_1 = 29 \cdot \frac{3}{10} = 8.7\\
    & \Rightarrow \xone_5 = \frac{9}{10}.\\\\
    &(w_4 - w_3) \cdot \yone_3 = 3 \cdot \frac{9}{10} = 2.7\\
    & (w_4 - w_2) \cdot \yone_2 = 4 \cdot \frac{7}{10} = 2.8\\
    &(w_4 - w_1) \cdot \yone_1 = 6 \cdot \frac{3}{10} = 1.8\\
    &\Rightarrow \xone_4 = \frac{7}{10}.\\\\
    &(w_3 - w_2) \cdot \yone_2 = 1 \cdot \frac{7}{10} = 0.7\\
    &(w_3 - w_1) \cdot \yone_1 = 3 \cdot \frac{3}{10} = 0.9\\
    &\Rightarrow \xone_3 = \frac{3}{10}.
\end{align*}
By the utility analysis above, arm $w_3, w_2, w_1, w_1, w_1$ are the favourite arms of $\gamma$-mean buyer with value $w_5, w_4, w_3, w_2, w_1$ respectively (and will be pulled almost always). Since in each round the arm's probability of allocation is the same, we can verify that $w_i$ will indeed get the item with probability $\xone_i + o(\gamma)$ if the buyers are running $\gamma$-mean based algorithms. The last thing we need to do is to verify that $\yone$ is feasible: that it is possible to give items with the probabilities specified by $\yone$. We know that a buyer is bidding $w_3, w_2, w_1$ with probability $1/5, 1/5, 3/5$ respectively, $\yone$ is a feasible interim allocation by border's constraint.
 
Similarly, let $\atwo$ be the auction where buyer gets item with probability $\ytwo_i$ by bidding $w_i$ in each round. 
\begin{align*}
    &{ (w_5 - w_4) \cdot \ytwo_4 = 23 \cdot \frac{9}{10} = 20.7}\\
    &(w_5 - w_3) \cdot \ytwo_3 = 26 \cdot \frac{7}{10} = 18.2\\
    &(w_5 - w_2) \cdot \ytwo_2 = 27 \cdot \frac{3}{10} = 8.1\\
    &(w_5 - w_1) \cdot \ytwo_1 = 29 \cdot \frac{3}{10} = 8.7\\
    & \Rightarrow \xtwo_5 = \frac{9}{10}\\\\
    &{ (w_4 - w_3) \cdot \ytwo_3 = 3 \cdot \frac{7}{10} = 2.1}\\
    &(w_4 - w_2) \cdot \ytwo_2 = 4 \cdot \frac{3}{10} = 1.2\\
    &(w_4 - w_1) \cdot \ytwo_1 = 6 \cdot \frac{3}{10} = 1.8\\
    &\Rightarrow \xtwo_4 = \frac{7}{10}\\\\
    &(w_3 - w_2) \cdot \ytwo_2 = 1 \cdot \frac{3}{10} = 0.3\\
    &{ (w_3 - w_1) \cdot \ytwo_1 = 3 \cdot \frac{3}{10} = 0.9}\\
    &\Rightarrow \xtwo_3 = \frac{3}{10}
\end{align*}
We can see that in $\atwo$, arm $w_4, w_3, w_1, w_1, w_1$ are the favourite arms of buyer with value $w_5, w_4, w_3, w_2, w_1$ respectively, which is consistent with $\xtwo$. Lastly, We know that a buyer is bidding $w_4, w_1$ with probability $1/5, 4/5$ respectively, thus $\ytwo$ is a feasible interim allocation by border's constraint.

Now let 
\begin{align*}
    & x = \frac{\xone + \xtwo}{2} = \begin{bmatrix}
        9/10\\
        7/10\\
        3/10\\
        3/10\\
        3/10
    \end{bmatrix}, 
    y = \frac{\yone + \ytwo}{2} = \begin{bmatrix}
        9/10\\
        9/10\\
        8/10\\
        5/10\\
        3/10\\
    \end{bmatrix}.
\end{align*}
Assume for contradiction that $(x, y)$ is $n$-buyer auction feasible, then there exists an auction $A$ where given clever $0$-mean buyer strategies $S_1, \ldots, S_n$, for all $i, j$, $X_{ij}^A (\dist, S_1, \ldots, S_n)$ $= x_j$ and $Y_{ij}^A (\dist, S_1, \ldots, S_n) = y_j$. we first find the preferred arms of buyer with value $w_i$ at the end of the auction. 

\begin{align*}
    &(w_5 - w_4) \cdot y_4 = 23 \cdot \frac{9}{10} = 20.7\\
    &{\bf (w_5 - w_3) \cdot y_3 = 26 \cdot \frac{8}{10} = 20.8}\\
    &(w_5 - w_2) \cdot y_2 = 27 \cdot \frac{5}{10} = 13.5\\
    &(w_5 - w_1) \cdot y_1 = 29 \cdot \frac{3}{10} = 8.7\\\\
    &{\bf (w_4 - w_3) \cdot \ytwo_3 = 3 \cdot \frac{8}{10} = 2.4}\\
    &(w_4 - w_2) \cdot \ytwo_2 = 4 \cdot \frac{5}{10} = 2\\
    &(w_4 - w_1) \cdot \ytwo_1 = 6 \cdot \frac{3}{10} = 1.8
\end{align*}

Thus towards the end of the auction, $w_3$ is the preferred arm by buyer with value $w_5$ and $w_4$. By Claim~\ref{claimSameBidAlloc} (with $n = 2, S = \{w_5, w_4\})$, arm $w_3$ cannot allocate the item with probability larger than $\frac{1}{2} (1 - (1 - 2/5)^2)/(2/5) = \frac{4}{5}$ towards the end of the auction. However by the same Claim~\ref{claimSameBidAlloc} (with $n = 2, S = \{w_5\})$, we know that any arm preferred by buyer with value $w_5$ cannot give the item with probability larger than $\frac{1}{2} (1 - (1 - 1/5)^2)/(1/5) = \frac{9}{10}$ in any round. Since buyers are supposed to receive the item with probability $x_5 = \frac{9}{10}$ on average when their value is $w_5$, in each round of auction $w_5$-buyer must get the item with probability exactly $\frac{9}{10}$, which contradicts the fact that $w_5$ gets item with probability $\leq \frac{4}{5}$ towards the end of the repeated auction. 
\end{proof}
\begin{claim}\label{claimSameBidAlloc}
Let $S$ be a subset of $\D$'s support. Let $q(S) = \sum_{v \in S} q(v)$ and let $b(S)$ be some fixed value. Let $A$ be a one round auction with $n$ symmetric buyers and $1$ item where buyers are constrained to bid exactly $b(S)$ when their value $v$ is in $S$. Then given that a buyer $i$ draws a value $v \in S$, they cannot receive the item with probability larger than $\frac{1}{n q(S)} (1 -  (1- q(S))^{n}).$ 
\end{claim}
\begin{proof}
Let event $A$ be that buyer $i$ gets the item; event $B$ be that buyer $i$ has value $v \in S$; event $C$ be that some buyer has value $v \in S$. Then the probability we want is $\Pr[A | B]$. Since event $B$ implies event $C$, 
\begin{align*}
    \Pr[A | B] = \Pr[ A | B, C] = \frac{\Pr[A, B| C]}{\Pr[B | C]} \leq \frac{\Pr[A| C]}{\Pr[B | C]}.
\end{align*}
We know that $\Pr[A | C] = \frac{1}{n}$, since information from event $C$ is symmetric to all buyers. We can directly calculate 
$$\Pr[B | C] = \frac{\Pr[B, C] }{\Pr[C]} = \frac{\Pr[B]}{\Pr[C]} = q(S)/(1 -  (1- q(S))^{n}).$$ 
We conclude that $\Pr[A | B] = \frac{1}{n} \cdot (1 -  (1- q(S))^{n})/q(S)$. 
\end{proof}
\section{Additional Discussion on n-buyer BMSW Constraints} \label{sec:dual}

In this section we will give a more in depth analysis of the linear programs that maximizes revenue (i.e.\ welfare - utility) given BMSW constraints and $n$-buyer BMSW constraints. The linear programs (we call them \SingleLP{} and \BorderLP) are presented below again for readers' convenience. We will call the optimal solution to the linear programs $\SingleOPT(\D)$ and $\BorderOPT(\D)$ respectively. 
\begin{align*}
    \textbf{Maximize } & \sum_{i = 1}^m q_i \cdot (v_i \cdot x_i - u_i)  \tag{\SingleLP}\\
    \textbf{Subject to } 
    & u_i \geq (v_i - v_j) \cdot x_j \hspace{2mm} \forall v_j < v_i, \\
    & u_i \geq 0 \hspace{2mm} \forall u_i,\\
    & 0 \leq x_i \leq 1 \hspace{2mm} \forall x_i,\\
    & \vec{x} \text{ is monotone increasing.} \\\\
    \textbf{Maximize } & \sum_{i = 1}^m q_i \cdot (v_i \cdot x_i - u_i)  \tag{\BorderLP} \\
    \textbf{Subject to } 
    & u_i \geq (v_i - v_j) \cdot x_j \hspace{2mm} \forall v_j < v_i, \\
    & u_i \geq 0 \hspace{2mm} \forall u_i,\\
    & \vec{x} \text{ satisfies Border's constraints for $n$ buyers,}\\
    & \vec{x} \text{ is monotone increasing.} 
\end{align*}

Firstly, We derive the Lagrangian relaxation of \SingleLP{} and \BorderLP. We provide three properties that uniquely determines an optimal solution to the Lagrangian relaxation in the single buyer setting. This perspective gives us additional intuition to the proof of theorem 3.5 in \cite{BravermanMSW18} (which establishes that for $\D$ supported on $[1, H]$, both $\max_{\mathcal{D}}\left(\val(\D)/\SingleOPT(\D)\right)$ and $\max_{\D}\left(\SingleOPT(\D)/\rev(\D)\right)$ tend to infinity as $H$ tend to infinity).  

Secondly, we attempt to generalize the three properties to the multiple buyer setting. While two of the properties still hold, we show that the third property no longer hold for \BorderLP{} induced by all distributions $\D$. We present a restricted class of distributions where the third property does hold. In this case, the three properties would uniquely determine the optimal solution to \BorderLP{} as well.



\subsection{Lagrangian Formulation} 
We start by constructing the following lagrangian relaxation of \SingleLP, which we call \textbf{\SingleRelaxLP}. 
\begin{align*}
    \mathbf{min}_{\lambda} \spc \mathbf{max}_{x}  \spc \spc \spc & \sum_{i=1}^m q_i (v_i x_i - u_i) \mathbf{dv} + \sum_{i=1}^m \sum_{j=1}^{i}  \big(u_i - (v_i - v_j)x_j\big)\cdot \lambda_{ij}  \tag{\SingleRelaxLP}\\
    \mathbf{subject \spc to \spc} \spc \spc \spc & u_i \geq 0 \spc \spc \spc\forall \spc i \in [m],\\
	& 0 \leq x_i\leq 1 \spc \spc \spc \forall i \in [m],\\
	&\vec{x} \text{ monotone.}
\end{align*}
By a) consolidating the $u_i$ and $x_i$ terms, b) rescaling the lagrangian multipliers by $\lambda_{ij} := \frac{\lambda_{ij}}{q_i}$, and c) observing that optimal $\lambda$ must result in the multiplier of $u_i$ being zero, we arrive at the following equivalent optimization problem: 
\begin{align*}
     \mathbf{min}_{\lambda} \spc \mathbf{max}_{x} \spc \spc \spc &\sum_{i=1}^m \Big(q_iv_i - \sum_{k=i}^m q_k (v_k - v_i) \lambda_{ki} \Big) \cdot x_i \tag{\SingleRelaxLP} \\
    \mathbf{subject \spc to \spc} \spc \spc \spc &u_i \geq 0 \spc \spc \spc\forall \spc i \in [m],\\
	& 0 \leq x_i\leq 1 \spc \spc \spc \forall i \in [m],\\
	&\vec{x} \text{ monotone, }\\
	&\sum_{j=1}^{i} \lambda_{ij} = 1 \spc \spc \spc \forall i \in [m].
\end{align*}

We make one last reformulation so that our objective can be viewed as the virtual welfare of the buyers. From now on, when we refer to the optimization problem \SingleRelaxLP, we will be referring to this final reformulation. 
\begin{align*}
     \mathbf{min}_{\lambda} \spc \mathbf{max}_{x} \spc \spc \spc &\sum_{i=1}^m q_i\cdot \phi(i, \lambda)  \cdot x_i \tag{\SingleRelaxLP}\\
    \mathbf{subject \spc to \spc} \spc \spc \spc &u_i \geq 0 \spc \spc \spc\forall \spc i \in [m],\\
	& 0 \leq x_i\leq 1 \spc \spc \spc \forall i \in [m],\\
	&\vec{x} \text{ monotone, }\\
	&\sum_{j=1}^{i} \lambda_{ij} = 1 \spc \spc \spc \forall i \in [m],\\
	&\phi(i, \lambda) = v_i - \sum_{k=i}^m \frac{q_k}{q_i} (v_k - v_i) \lambda_{ki} \spc \spc \spc \forall i \in [m]. 
\end{align*}

Finally, we observe that the process of create the lagrangian relaxation for \BorderLP{} is identical to that for \SingleLP. Let us call the resulting optimization problem \BorderRelaxLP. We note that  the only difference between \SingleRelaxLP{} and \BorderRelaxLP{} is that $\vec{x}$ must satisfy border constraints in \BorderRelaxLP.  


\subsection{Three Properties of the Lagrangian Variable}  \label{subsec:threeProp}
Let us first see a useful property of the virtual value $\phi$: that there always exists an optimal $\lambda$ such that $\phi(i, \lambda)$ is monotone in $i$. We use $\phi_i$ to denote $\phi(i, \lambda)$ when the $\lambda$ in question is clear.   
\begin{theorem}
    \label{phi_monotonicity}
    For every optimal $\lambda$, $\phi_i \leq \phi_j$ when $i < j$.
\end{theorem}
\textbf{Proof Sketch.} We can view the lagrangian relaxation as a two player game where the first player sets the $\lambda$ that minimizes the objective, and the other sets the $\vec{x}$ that maximizes the objective given the $\lambda$ the first player sets. Consider when the second player gets a $\lambda$ where $\phi_i > \phi_j$ for some $i < j$, it would want to set $x_i$ to be as large as possible and $x_j$ to be as little as possible. Unfortunately $\vec{x}$ has to be monotone. In this case the best course of action would be to set $x_i = x_j$. Now, consider the actions of the first player, the first player knows that the second player will set $x_i = x_j$ as long as they set $\phi_i > \phi_j$. Now they observe that for any $\lambda_{kj} > 0$, one unit of $\lambda_{kj}$ lowers the objective by $q_k (v_k - v_j) x_j$, and we know that one unit of $\lambda_{ki}$ lowers the objective more because $q_k (v_k - v_i) x_i \geq q_k (v_k - v_i) x_j$. As long as $\phi_i > \phi_j$, the second player can just decrease $\lambda_{kj}$ (which increases $\phi_j$) and increase $\lambda_{ki}$ (which decreases $\phi_i$) to make the objective even smaller. As a result, the second player decides it's not a good idea to make $\phi_i > \phi_j$. A full proof of theorem~\ref{phi_monotonicity} can be found in section~\ref{subsec:formal}. 

Now we are ready to discuss the following three properties that in combination determines a unique and optimal $\lambda$ for \SingleRelaxLP. 

\begin{property}
\label{sprop1}
    $\phi(i, \lambda) \geq 0$ for all $i$. 
\end{property}

\begin{property}
\label{sprop2}
    Let $g(i, \lambda) = \{j | \lambda_{ij} > 0\}$. Then $\min{g(i, \lambda)} \geq \max{g(i - 1, \lambda)}$ for all $i$. Equivalently, there are no $i < j < k < l$ such that $\lambda_{li} > 0$ and $\lambda_{kj} > 0$.
\end{property}

\begin{property}
\label{sprop3}
    $\lambda_{ki} = 0$ for all $k, i$ such that $\exists j < i$ where $q_j \cdot \phi(j, \lambda) > 0$.
\end{property}
\begin{lemma}
\label{scond1}
    There exists an optimal $\lambda$ such that \ref{sprop1} is true. 
\end{lemma}
\begin{lemma}
\label{scond2}
    There exists an optimal $\lambda$ such that Property \ref{sprop2} is true.
\end{lemma}
\begin{lemma}
\label{scond3}
    Every optimal $\lambda$ satisfies \ref{sprop3}.
\end{lemma}
\begin{theorem}\label{thm:threePropExist}
    The $\lambda$ that satisfies properties \ref{sprop1},\ref{sprop2},\ref{sprop3} exists.
\end{theorem}
\textbf{Proof sketch.} 
\begin{enumerate}
    \item 
    There exists an optimal $\lambda$ that satisfies property~\ref{sprop1}. Consider the same analog with two players setting $\lambda$ and $\vec{x}$. Assume player one is setting an optimal $\lambda$ where $\phi_i < 0$ for some $i$. Then by theorem~\ref{phi_monotonicity}, $\phi_j < 0$ for all $j < i$. Now consider player two who knows that the virtual values $\phi_1$ to $\phi_i$ are all negative. The best strategy should be to set $x_1$ to $x_i$ to be $0$. Knowing the strategies of player two, player one can decrease $\lambda_{ki}$ by $\epsilon$ without increasing the objective at all. Hence eventually player one can arrive at another optimal $\lambda$ where $\phi(i, \lambda) \geq 0$ for all $i$.  
    \item 
    There exists an optimal $\lambda$ that satisfies property~\ref{sprop1} and~\ref{sprop2}. Consider an optimal $\lambda$ that satisfies property ~\ref{sprop1} but for some $i < j < k < l$ where $\lambda_{li} > 0$ and $\lambda_{kj} > 0$. Now consider decreasing $\lambda_{li}$ by $\delta$ and increasing $\lambda_{lj}$ by delta. Next, increase $\lambda_{ki}$ and decrease $\lambda_{kj}$ by an amount $\delta'$ so that $\phi_i$ is the same as before any change. One can verify that $\phi_j$ has decreased after the change while no other $\phi$ has changed, which means the $\lambda$ after change is also optimal. If this $\lambda$ violates property~\ref{sprop1}, we can simply decrease $\lambda_{lj}$ and increase $\lambda_{ll}$ until $\phi_j$ is non negative. 
    \item 
    Every optimal $\lambda$ satisfies property~\ref{sprop3}. Consider player one setting an optimal $\lambda$ where $\lambda_{ki} > 0 $ for some $k, i$ where $\exists j < i$ such that $q_j \cdot \phi(j, \lambda) > 0$. By theorem~\ref{phi_monotonicity}, all $\phi_l$ where $l \geq j$ will be possible, and therefore player two should set $x_l = 1$ for all $l \geq j$. Now, player two's strategy, player one can simply decrease $\lambda_{ki}$ and increase $\lambda_{kj}$ by the same amount in order to lower the objective, which contradicts with $\lambda$ being optimal.  
\end{enumerate}
\begin{theorem}\label{thm:threePropUnique}
\label{unique}
    The $\lambda$ that satisfies properties \ref{sprop1},\ref{sprop2},\ref{sprop3} is unique.
\end{theorem}
\textbf{Proof Sketch.}
Consider what properties \ref{sprop1},\ref{sprop2},\ref{sprop3} is saying. Let us say ``$k$ is lowering index $i$" when $\lambda_{ki} > 0$.  Property \ref{sprop2} says that larger $k$ has to lower larger $i$. Property \ref{sprop1} and \ref{sprop3} basically says that if $i$ is the smallest number where $\phi_i > 0$, then $\phi_j = 0$ for all $j < i$. Moreover, all $k \in [m]$ has to be lowering some $j \leq i$. This essentially means that if one knows the distribution of $\lambda_{1}, \lambda_{2}... \lambda_{k-1}$, there is only one way to set $\lambda_{k}$: lower the smallest $i$ where $\phi_i > 0$. if $\phi_i$ becomes $0$, then lower $i+1$ next, etc. Hence one can expect the $\lambda$ that satisfies the three properties to be exactly determined by the following algorithm.   

\begin{algorithm}
  \caption{FillLowToHigh():}
  \label{lowToHigh}
\begin{algorithmic}[0]
    \STATE $\forall k\geq i, \lambda_{ki} \leftarrow 0$
    \FOR{k = 1, 2, ... m}
        \STATE $\lambda(k) \leftarrow 1$
        \WHILE{$\lambda(k) > 0$}
        \STATE $i \leftarrow \arg\min_{j}\{\phi(j, \lambda) > 0\}$
        \STATE $\lambda(\text{fill}) \leftarrow $ the value that if assigned to $\lambda_{ki}$ will make $\phi(i, \lambda) = 0$
        \STATE $\lambda_{ki} \leftarrow \min\big(\lambda(k), \lambda(\text{fill})\big)$
        \STATE $\lambda(k) \leftarrow \lambda(k) - \lambda_{ki}$
        \ENDWHILE
    \ENDFOR
    \STATE return $\lambda$
\end{algorithmic}
\end{algorithm}
The formal proofs of theorem~\ref{thm:threePropExist} and \ref{thm:threePropUnique} can be found at section~\ref{subsec:formal}.  

The existence of such a unique optimal $\lambda$ makes it much easier for us to understand the optimal solution to LP Single, and can lead to additional intuition. For instance, consider the following theorem in~\cite{BravermanMSW18}.   
\begin{theorem}[\cite{BravermanMSW18}]\label{thm:LPRevForER} 
Let distribution $\mathcal{D^*}$ be the equal revenue curve truncated at $H$ ($\mathcal{D^*}$ is supported on $[1, H]$), which means that the cumulative distribution function of $\mathcal{D^*}$ is
$$F(v) = 1 - 1/v \quad \forall v \in [1, H); \quad F(H) = 1.$$
Then \SingleOPT($\mathcal{D^*}$) = $\Theta(\log\log H)$, while $\val_n(\mathcal{D^*}) = \Theta(\log H)$.
\end{theorem}

\cite{BravermanMSW18} proves that \SingleOPT($\mathcal{D^*}$) = $\Theta(\log\log H)$ by first setting specific $\vec{x}$ and finding the corresponding objective value to prove \SingleOPT($\mathcal{D^*}$) $\geq  \log(\log H + 1) - 1$, then proving a $\log(\log H + 1)$ upper bound on \SingleOPT. By using the three properties of the optimal $\lambda$ and associated $\phi(i ,\lambda)$, we can actually prove that this upper bound is tight.

\begin{proof}[(alternate proof of Theorem~\ref{thm:LPRevForER}~\cite{BravermanMSW18})]

Let $(x^*, u^*)$ be the optimal solution to \SingleLP given distribution $\mathcal{D}$. 

In the discrete support setting we know there is a unique optimal $\lambda=\lambda^*$ that satisfies property~\ref{sprop1},~\ref{sprop2} and~\ref{sprop3}. Since $\phi(i, \lambda)$ is none negative and monotonely increases in $i$, setting $x_i = 1$ for all $i$ maximizes the objective. Therefore
$$\min_{\lambda}\max_x \sum_{i=1}^m q_i \phi(i, \lambda) x_i = \max_x \sum_{i=1}^m q_i \phi(i, \lambda^*) x_i = \sum_{i=1}^m q_i \phi(i, \lambda^*).$$ Let $w(v_i)$ be the largest value $v_j$ where $\lambda^*_{ji} > 0$, by property~\ref{sprop2} we know $w(v_i)$ is monotone increasing in $i$. Let $g(v_i)$ be the set of values $v_j$ where $\lambda^*_{ij} > 0$. By complementary slackness of $(\xst, \ust)$ and $\lambda^*$, it is equivalent to define $g(v_i)$ as the set of value $v_j$ where $u_i = (v_i - v_j)\xst_j$. 
By property~\ref{sprop3}, $\phi(i, \lambda^*) = 0, \forall i < \max(g(v_m))$. Then it must be the case that for all index $l < \max(g(v_m))$,
\begin{align}
&\sum_{i=1}^{l} q_i \phi(i, \lambda^*) = 0 \nonumber\\
\Rightarrow & \sum_{i=1}^{l} \left(q_i v_i - \sum_{k=i}^m q_k (v_k - v_i) \lambda_{ki}\right) = 0\nonumber\\
\Rightarrow & \sum_{i=1}^{l} q_i v_i = \sum_{i=1}^{l} \sum_{k=i}^m q_k (v_k - v_i) \lambda_{ki} = \sum_{k=1}^{w(l)} \sum_{i \in g(v_k)} q_k (v_k - v_i) \lambda_{ki}. \label{eq:phiZeroConservation}
\end{align}
Now consider $D^*_\epsilon$: a discretization of $\mathcal(D^*)$ where the support is $\{1, 1+ \epsilon ... H\}$ and $F(v) = 1 - \frac{1}{v}$ for all $v$ in the support. $\max(g(v)) - \min(g(v))$ can be made arbitrarily small for all $v$ by sufficiently small $\epsilon$, in other words, $\lim_{\epsilon \rightarrow 0} \max(g(v)) - \min(g(v)) = 0$. Also notice that for any $\delta$, there exists an $\epsilon$ where $g(v) \neq g(v + \delta)$ by $\phi$ monotonicity. Thus for $\mathcal{D} = \lim_{\epsilon \rightarrow 0} \mathcal{D}^*_{\epsilon}$, $g(v)$ just contains a point and $g(v) \neq g(w)$ for all $v \neq w$, which means that $g$ can be viewed as a strictly increasing function and $w$ can be viewed as $g^{-1}$. Thus in the continuous setting the analogous equation to equation~\ref{eq:phiZeroConservation} would be
\begin{align*}
    & \forall k \in supp(\mathcal{D}): \int_i^k f(v)(v - g(v))\spc dv = \int_1^{g(k)} f(v) v \spc dv \\
    & \Rightarrow \int_1^k f(v)(v - g(v)) \spc dv = \int_1^{k} f(g(v))g(v) g'(v) dv\\
    & \Rightarrow \int_1^k f(v)(v - g(v)) - f(g(v))g(v) g'(v) \spc dv = 0.
\end{align*}
Therefore $f(v)(v - g(v)) - f(g(v))g(v) g'(v) = 0$. Given the equal revenue curve, $f(v) = \frac{1}{v^2}$. Thus 
\begin{align*}
    &\frac{1}{v^2}(v - g(v)) - \frac{1}{g(v)^2} g(v) g'(v) = 0\\
    &\Rightarrow 1/v - g(v)/v^2 - \frac{g'(v)}{g(v)} = 0.
\end{align*}
As calculated in~\cite{BravermanMSW18}, the solution to above equation is $g(v) = \frac{v}{\log v + 1}$. Therefore $g(H) = \frac{H}{\log H + 1}$, and 
\begin{align*}
    \text{\SingleOPT} = \int_{g(H)}^H f(v)\phi(v, \lambda^*) = \int_{g(H)}^H f(v) v \spc dv = \int_{H/(\log H + 1)}^H \frac{1}{v^2} \cdot v \spc dv =  \log(\log H + 1).
\end{align*}
\end{proof}

\subsection{Extension to Multiple buyers}

Using the same proof logic as for the single buyer case, it is easy to see that theorem~\ref{phi_monotonicity} still holds and there is an optimal lambda which satisfies property~\ref{sprop1} and \ref{sprop2}. However, it is possible that no optimal $\lambda$ satisfies property~\ref{sprop3} due to the additional restriction that $x$ must satisfy Border's constraints, which makes our understanding of the solution to \BorderLP{} much more difficult. We can see this from the following example. 
\begin{remark}
    \label{counterexample}
    There exists an optimal $\lambda$ for some distribution $D$ where $\lambda_{ki} > 0$ for some $k, i$ such that $\exists j < i$ where $x_j$ has a positive multiplier.
\end{remark}
\begin{proof}
    We provide an example for which this statement is true. Consider the following distribution:
    $$q = \left[\frac{1}{4}, \frac{1}{4}, \frac{1}{4}, \frac{1}{4}\right], v = [1, 9, 10, 15]$$
    Then in order for Property \ref{sprop3} to hold, it must be true that $\phi_1 \leq 0, \phi_3 = v_3$ and $\phi_4 = v_4$ because no $\lambda$ has the result of $\phi_2 = 0$. Then every $\lambda$ which satisfies \ref{sprop3} has the same objective value. One such $\lambda$ which satisfies \ref{sprop3} is: $\lambda_{11} = 1, \lambda_{21} = \frac{1}{8}, \lambda_{22} = \frac{7}{8}, \lambda_{32} = 1,$ and $\lambda_{42} = 1$. The resulting $\phi$ is strictly increasing, so the optimal $x^*$ is for Border's constraint to be tight. Thus, $x^* = [\frac{1}{8}, \frac{3}{8}, \frac{5}{8}, \frac{7}{8}]$. Now consider an alternate $\lambda'$ such that $\lambda_{11} = 1, \lambda_{21} = 1, \lambda_{32} = \frac{3}{11},$ and $\lambda_{42} = \frac{8}{11}$. Then the resulting $\phi'$ is also strictly increasing, so the optimal $x^*$ is the same as for $\lambda$. The difference in objective between $\lambda$ and $\lambda'$ is:
    $$\lambda - \lambda' = \frac{1}{4}(15 - 9)(1)\frac{7}{8} - \frac{1}{4}(15 - 9)\left(\frac{3}{11}\right)\frac{7}{8} - \frac{1}{4}(15 - 10)\left(\frac{8}{11}\right)\frac{7}{8} > 0.$$
    Because $\lambda'$ yields a better objective, we know $\lambda$ must not be optimal and so the optimal $\lambda^*$ for this distribution must not satisfy Property \ref{sprop3}.
\end{proof}

A natural question following this discovery is: under what conditions could we generalize the results for single buyer setting to multiple buyer setting? Here we prove that for a restricted class of distribution $\D$ which satisfy $\frac{f(v)}{F(v)} \leq \frac{1}{H -v}$ for all $v$, theorem~\ref{thm:threePropExist} and~\ref{thm:threePropUnique} are generalizable. 
\begin{theorem}
\label{sufficiency}
    If the distribution satisfies $\frac{f(v)}{F(v)} \leq \frac{1}{H - v}$ for all $v$ where $H$ is the maximum value possible, then there is a unique optimal $\lambda$ for the multi-buyer LP which satisfies Properties \ref{sprop1}, \ref{sprop2},  and \ref{sprop3}.
\end{theorem}
\begin{proof}
    If $(v_k - v_j)x_j \geq (v_k - v_i)x_i$ for all $i \geq j$ and $k \geq i$, then we can extend Lemma \ref{scond3} to the multi-buyer case because we will have normalized for $x$. Lemmas \ref{scond1} and \ref{scond2} were already extendable, so all three lemmas hold true for the multi-buyer case under a distribution of the specified type and we can also apply the logic in Theorem \ref{unique}. Note that for any $\lambda$ which induces a monotone $\phi$ and satisfies Property \ref{scond1}, it is optimal for Border's constraints to be tight because it is optimal to maximize $x_i$ for every $i$. For any minmax problem $(\min_\lambda \max_x f(\lambda,x))$, if $x^*$ is optimal for the inner maximization \emph{for all $\lambda$}, then the problem is equivalent to $\min_\lambda f(\lambda, x^*)$. In our problem, we know that $x_j = F(v_j)$ when Border's constraints are tight. Therefore $x^*_j = F(v_j)$. We want:
    $$(v_k - v_j)x_j \geq (v_k - v_i)x_i \implies \frac{F(v_j)}{F(v_i)} \geq \frac{v_k - v_i}{v_k - v_j}.$$
    Notice that $\frac{v_k - v_i}{v_k - v_j}$ is maximized when $k$ is set to the maximum $v$ possible. Let $H$ be this maximum $v_k$. Then $\frac{F(v_j)}{F(v_i)} \geq \frac{H - v_i}{H - v_j}$. Now we can compute:
    \begin{align*}
        f(v) &= \lim_{\epsilon \rightarrow 0} \frac{F(v + \epsilon) - F(v)}{\epsilon} \\
        &= \lim_{\epsilon \rightarrow 0} \frac{F(v + \epsilon)}{\epsilon} \cdot \left(1 - \frac{F(v)}{F(v + \epsilon)}\right) \\
        &\leq \lim_{\epsilon \rightarrow 0} \frac{F(v + \epsilon)}{\epsilon} \cdot \left(1 - \frac{H - v - \epsilon}{H - v}\right) \\
        &= \lim_{\epsilon \rightarrow 0} \frac{F(v + \epsilon)}{H - v} \\
        &= \frac{F(v)}{H - v},
    \end{align*}
    which implies $\frac{f(v)}{F(v)} \leq \frac{1}{H - v}$, as desired.
\end{proof}
\begin{corollary}
    For any distribution which satisfies $\frac{f(v)}{F(v)} \leq \frac{1}{H - v}$, then $F(v) \geq \frac{1}{H - v}$ for all $v$. 
\end{corollary}
\begin{proof}
    Proof by contradiction. Let $F^*(v) = \frac{1}{H - v}$. Assume $F(v) < F^*(v)$ for some $v$. Let $y$ be the smallest value greater than $v$ such that $F(y) = F^*(y)$. Note that $y$ must exist by the Intermediate Value Theorem. Then $\int_v^y f(x) dx > \int_v^y f^*(x) dx$. By the Mean Value Theorem, there must exist some $z \in [v, y]$ such that $f(z) > f^*(z)$. Because of the minimality of $y$, it also must be true that $F(z) \leq F^*(z)$. This implies $\frac{f(z)}{F(z)} > \frac{f^*(z)}{F^*(z)}$, which is a contradiction to 
    $$\frac{f(z)}{F(z)} \leq \frac{f^*(z)}{F^*(z)} = \frac{(H-v)^{-2}}{(H-v)^{-1}} = \frac{1}{H - v}.$$.
\end{proof}

\subsection{Omitted Formal Proofs in Section~\ref{subsec:threeProp}} \label{subsec:formal} 

\begin{lemma}
\label{x_equal}
For any fixed $\lambda$ with $\phi(i, \lambda) \geq \phi(i + 1, \lambda)$ for some $i$, there exists an optimal $x$ with the following property: $x_i = x_{i + 1}$. 
\end{lemma}
\begin{proof}
Proof by contradiction. Assume that $x_i \neq x_{i + 1}$ but $\phi(i ,\lambda) \geq \phi(i + 1, \lambda)$ for some optimal $x$. In order for some $x$ to be feasible for a distribution, the probability the item is allocated must be $\leq 1$. 
Note that:
\begin{align*}
    Pr[\text{ item is allocated }] &= \sum_x \sum_{i = 1}^m q(v_i) \cdot x_i \\
    &= n \sum_{i = 1}^m q(v_i) \cdot x_i.
\end{align*}
Suppose we know that there exists some feasible $x$ for some specific distribution. Then one way to design a feasible $x'$ with every index the same except for a specific $x_i$ and $x_{i + 1}$ is to have $q(v_i) \cdot x_i + q(v_{i + 1}) \cdot x_{i + 1} = q(v_i) \cdot x'_i + q(v_{i + 1}) \cdot x'_{i + 1}$, because this will guarantee that the item is allocated with probability $\leq 1$. Furthermore, because $x$ is monotone, we know that $x_i \leq x_{i + 1}$.  Now consider $x'$ which is equal to $x$ except that $x'_i = x'_{i + 1} = \frac{q(v_i)x'_i + q(v_{i + 1})x'_{i + 1}}{q(v_i) + q(v_{i + 1})}$ for the previously mentioned $i$.  Recall that our objective is:
$$\min_\lambda \max_{x} \sum_{i = 1}^m q(v_i) \cdot \phi(i, \lambda) \cdot x_i.$$
 If we consider the difference in the objective between $x_i'$ and $x_i$, we find:
\begin{align*}
    &q(v_i)\phi(i ,\lambda)(x'_i - x_i) + q(v_i + 1)\phi(i+1, \lambda)(x'_{i + 1} - x_{i + 1}) \\ 
    \geq &\phi(i+1,\lambda) \left[(q(v_i) + q(v_{i + 1}))x'_i - q(v_i)x_i - q(v_{i + 1})x_{i + 1}\right] \\
    = &\phi(i+1, \lambda)\left[(q(v_i)x_i + q(v_{i + 1})x_{i + 1} - q(v_i)x_i - q(v_{i + 1})x_{i + 1}\right] \\
    =&0.
\end{align*}
This means we can improve the objective, which is a contradiction and thus implies $x_i$ must equal $x_{i + 1}$.
\end{proof}

\begin{proof}[Proof of Theorem~\ref{phi_monotonicity}]
Let $OBJ(\lambda, x)$ be the value of the objective function with specific inputs $\lambda$ and $x$ and $OBJ(\lambda)$ be the optimal value of the objective function given a specific $\lambda$.

Proof by contradiction. Assume theorem $\ref{phi_monotonicity}$ is not true. Then there exists an optimal solution $OBJ(\lambda')$ with some $i$ such that $\phi_i > \phi_{i + 1}$. Now consider a $OBJ(\lambda'')$ such that $\phi_i = \phi_{i + 1}$ which is constructed by redirecting flow in $OBJ(\lambda)$ from $\phi_{i + 1}$ to $\phi_i$. Then $OBJ(\lambda'')$ is identical to $OBJ(\lambda')$ except at $\phi_i$ and $\phi_{i + 1}$. Note that when redirecting $\epsilon$ flow from $\phi_{i + 1}$ to $\phi_i$, the change in our objective is:
$$\epsilon \cdot (-q(v_k)(v_k - v_i)(x_i) + q(v_k)(v_k - v_{i + 1})(x_{i + 1})),$$
where $k$ is some index greater than $i + 1$ and $\lambda_{k(i + 1)} > 0$. Further note that $k$ is guaranteed to exist because originally $\phi_{i + 1} = v_{i + 1}$ and $v_{i + 1} > v_{i}$, and so a non-monotonicity could not occur if no such $k$ existed.
From $\ref{x_equal}$ we know that for every $\lambda$ there exists a subclass of optimal $x$s such that $x_i = x_{i + 1}$ for any $\phi_i \geq \phi_{i + 1}$. Choose $x'$ to be one such optimal $x$ for $\lambda'$ and $x''$ to be another such optimal $x$ for $\lambda''$. We can observe that $$(\lambda'', x'') < (\lambda', x'')$$ because when $x_i = x_{i + 1}$, any redirected flow changes the objective by $x_i \cdot \epsilon \cdot (q(v_k)(v_i - v_{i + 1}))$ which is strictly negative. Therefore:
$$OBJ(\lambda'') = OBJ(\lambda'', x'') < OBJ(\lambda', x'') \leq OBJ(\lambda', x') = OBJ(\lambda'),$$
where the $\leq$ is because $x'$ is optimal for $\lambda'$. This is a contradiction because $OBJ(\lambda')$ was optimal.
\end{proof}

\begin{proof}[Proof of Lemma~\ref{scond1}]
    Proof by contradiction. Assume Lemma \ref{scond1} does not hold. Because $q(v_i)$ can never be negative, it must be the case that $\phi(i, \lambda)$ is negative for some number of $i$. Choose an optimal $\lambda$ which minimizes the number of $i$ such that $\phi(i, \lambda) < 0$. Consider the largest value $i$ such that $\phi(i, \lambda) < 0$ for the chosen $\lambda$. We first show that it must be true that $x_i = 0$, by induction. By Theorem \ref{phi_monotonicity}, all $\phi_j$ are also negative for $j < i$. Our base case is that $x_1 = 0$. We know that $q(v_1) \cdot \phi(1, \lambda) \leq 0$, and therefore $q(v_1) \cdot \phi(1, \lambda) \cdot x_1$ is maximized at $x_1 = 0$. Now assume that $x_k = 0$ for some $k \leq i$. We show that $x_{k + 1} = 0$ as well. Note that $x_{k + 1}$ is maximized at the lowest $x_{k + 1}$ possible satisfying monotonicity constraints, because $q(v_{k + 1}) \cdot \phi(k + 1, \lambda) \leq 0$. Because $x_k = 0$, the lowest $x_{k + 1}$ is also $0$. Therefore, $x_i = 0$. 
    
    Now consider the set $S$ of all $s$ in increasing order such that $\lambda_{si} > 0$. Recall that:
    $$\phi(i, \lambda) = \frac{q(v_i)v_i
-  \sum_{k = i}^m q(v_k) \cdot (v_k - v_i) \lambda_{ki}}{q(v_i)}.$$
    Note that if we decrease $\lambda_{si}$ for any $s \in S$, $\phi(i, \lambda)$ always increases. Further note that because the first term in the numerator of $\phi(i, \lambda)$ is non-negative, if $\phi(i, \lambda)$ is negative then there must exist some $\lambda_{si}$ such that $\lambda_{si} > 0$. Construct $\lambda'$ which is the same as $\lambda$ but with the following changes: in increasing order of $s \in S$, let $\lambda'_{si} = 0$ and $\lambda'_{ss} = \lambda_{si}$ until $\phi(i, \lambda)$ is positive. Then for the largest $s$ such that flow was redirected, add $a$ to $\lambda'_{si}$ and subtract $a$ from $\lambda'_{ss}$ where $a$ is exactly amount of additional $\lambda'_{si}$ needed to make $\phi(i, \lambda') = 0$. These transformations maintain the property that $\sum_{l = 0}^s \lambda_{sl} = 1$ for each $s$, as desired. Reducing $\lambda_{si}$ does not change the objective because $x_i = 0$, as established earlier. Increasing $\lambda_{ss}$ also does not change the objective because $q(v_s) \cdot (v_s - v_s) \cdot \lambda_{ss} = 0$. Therefore the value of the objective is the same for $\lambda$ and $\lambda'$, and $\lambda'$ must be optimal if $\lambda$ is optimal. However, $\phi_i$ is now non-negative, which contradicts our assumption that $\lambda$ minimized the number of $i$ such that $\phi(i, \lambda) < 0$. Therefore must exist an optimal $\lambda$ such that Property \ref{sprop1} is true. 
\end{proof}
\begin{proof}[Proof of Lemma~\ref{scond2}]
    Proof by contradiction. Assume this is not true. Then there is some optimal $\lambda$ with $i < j < k < l$ such that $\lambda_{li} > 0$ and $\lambda_{kj} > 0$. Let $\lambda'$ be the same as $\lambda$ except for the following changes: $\lambda'_{kj} = \lambda_{kj} - \delta, \lambda'_{ki} = \lambda_{ki} + \delta, \lambda'_{li} = \lambda_{li} - \epsilon$, and $\lambda'_{lj} = \lambda_{lj} + \epsilon$ with $\delta > 0$ and $\epsilon = \frac{\delta \cdot q(v_k) \cdot (v_k - v_i)}{q(v_l) \cdot (v_l - v_i)}$. Then the change in $\phi_i$ is:
    \begin{align*}
            &-\delta \cdot q(v_k) \cdot (v_k - v_i) + \epsilon \cdot q(v_l) \cdot (v_l - v_i) \\
            = & -\delta \cdot q(v_k) \cdot (v_k - v_i) + \frac{\delta \cdot q(v_k) \cdot (v_k - v_i)}{q(v_l) \cdot (v_l - v_i)} \cdot q(v_l) \cdot (v_l - v_i) \\
            = & 0. 
    \end{align*}
    The change in $\phi_j$ is:
    \begin{align*}
        &\delta \cdot q(v_k) \cdot (v_k - v_j) - \epsilon \cdot q(v_l) \cdot (v_l - v_j)  \\
        = & \delta \cdot q(v_k) \cdot (v_k - v_j) - \delta \cdot q(v_k) \cdot  \frac{(v_k - v_i)}{(v_l - v_i)} \cdot (v_l - v_j)\\
        \leq & 0.
    \end{align*}
    where the last line is because $\frac{v_k - v_i}{v_l - v_i} \cdot \frac{v_l - v_j}{v_k - v_j}$ is always greater than 1. Note that when the multiplier of $x_l$ is non-negative for any $l$ it is optimal to set $x_l = 1$, so $x_i = x_j = 1$. Therefore our objective is lower for $\lambda'$ than for $\lambda$ when $\phi_j > 0$. Otherwise,  $x_i = x_j = 0$ and the objective for $\lambda$ is the same as for $\lambda'$. Thus if $\lambda$ was optimal, $\lambda'$ must be optimal as well, a contradiction.
\end{proof}
\begin{corollary}
    \label{scond2pos}
    If $\lambda$ has $i < j < k < l$ such that $\lambda_{li} > 0$ and $\lambda_{kj} > 0$, then $\lambda$ is not optimal when $\phi_j > 0$. 
\end{corollary}
\begin{proof}[Proof of Lemma~\ref{scond3}]
    Proof by contradiction. Assume this is not true. Then there is some optimal $\lambda$ with $j < i < k$ such that $\lambda_{ki} > 0$ but $x_j$ has a positive multiplier. Construct $\lambda'$ which is the same as $\lambda$ but with the following changes: $\lambda'_{kj} = \lambda_{kj} + \delta$ and $\lambda'_{ki} = \lambda_{ki} + \delta$ where $\delta$ is $\min ($the amount needed to lower $\phi_j$ to $0, \lambda_{ki})$.  Note that when the multiplier of $x_l$ is non-negative for any $l$ it is optimal to set $x_l = 1$, so $x_i = x_j = 1$. Then the change in objective from $\lambda$ to $\lambda'$ is:
    $$\delta q(v_k) (v_k - v_i) - \delta q(v_k) (v_k - v_j) < 0.$$
   Therefore our objective is lower for $\lambda'$ than for $\lambda$ and thus $\lambda$ must not have been optimal, a contradiction.
\end{proof}
\begin{proof}[Proof of Theorem~\ref{thm:threePropExist}]
    Consider the set of optimal $\lambda$ which satisfy Property \ref{sprop2}. Choose the $\lambda$ within this set which minimizes the number of $i$ such that $\phi(i, \lambda) < 0$. Then $\lambda$ necessarily satisfies Property \ref{sprop3} as well. We prove that there is an optimal $\lambda'$ which satisfies Properties \ref{sprop1},\ref{sprop2},\ref{sprop3} and will find it by adjustments to $\lambda$. Let $\lambda' = \lambda$. If $\lambda'$ satisfies \ref{sprop1}, then we are done. Otherwise, we shall make modifications to $\lambda'$ very similar to those done in the proof of \ref{scond1}. 
    
    Proof by contradiction. Consider the largest value $i$ such that $\phi(i, \lambda) < 0$ for the chosen $\lambda$ and the set $S$ of all $s$ in decreasing order such that $\lambda_{si} > 0$. Make the following changes to $\lambda'$: in decreasing order of $s \in S$, let $\lambda'_{si} = 0$ and $\lambda'_{s(i + 1)} = \lambda_{si}$ until $\phi(i, \lambda)$ is positive. Then for the largest $s$ such that flow was redirected, add $a$ to $\lambda'_{si}$ and subtract $a$ from $\lambda'_{s(i + 1)}$ where $a$ is exactly amount of additional $\lambda'_{si}$ needed to make $\phi(i, \lambda') = 0$. Reducing $\lambda'_{si}$ does not change the objective because $x_i = 0$. If $\phi_{i + 1}$ was non-positive before, then increasing $\lambda'_{s(i + 1)}$ does not change the objective because $x_j = 0$ as well. Otherwise $\phi_{i + 1}$ was positive, so increasing $\lambda'_{s(i + 1)}$ strictly decreases the objective.
    
    Because $\lambda$ satisfied Property \ref{sprop2} before, and we only redirected flow in decreasing order to $\phi_{i + 1}$, it is impossible that we are now violating \ref{sprop2}. Rigorously, there were no $i < j < k < l$ such that $\lambda_{li} > 0$ and $\lambda_{kj} > 0$. Consider the largest $s$ such that $\lambda_{si} > 0$ and it changes in $\lambda'$. The only way it is possible to now have an $i < j < k < l$ such that $\lambda'_{li} > 0$ and $\lambda'_{kj} > 0$ is if there is some $t > s$ and $r \leq i$ such that $\lambda_{tr} > 0$ (for the original $\lambda$). This is impossible because if $r < i$, Property \ref{sprop2} would have been violated in our original $\lambda$, and if $r = i$, this would violate the maximality of $s$. We can use this line of reasoning for every $s$ where $\lambda_{si} \neq \lambda'_{si}$ to show that $\lambda'$ must also satisfy Property \ref{sprop2}. Similarly, we cannot now be violating \ref{sprop3} because by Theorem \ref{phi_monotonicity} $\phi_j \leq 0$ for all $j \leq i$ and the only index $k$ for which $\lambda'_{lk} > \lambda_{lk}$ for any $l$ is $k = i + 1$.
    
    Therefore the value of the objective for $\lambda'$ is less than or equal to that of $\lambda$, so $\lambda'$ must be optimal if $\lambda$ is optimal. However, $\phi_i$ is now non-negative, which contradicts our assumption that $\lambda$ minimized the number of $i$ such that $\phi(i, \lambda) < 0$. Therefore must exist an optimal $\lambda$ that satisfies Properties \ref{sprop1},\ref{sprop2},\ref{sprop3}.
\end{proof}
\begin{proof}[Proof of Theorem~\ref{thm:threePropUnique}]
    We first propose an algorithm to generate a $\lambda$ which satisfies Properties \ref{sprop1}, \ref{sprop2}, and \ref{sprop3}. The algorithm is as follows: Begin with $\lambda_{ss} = 1$ for all values $s$. For each $k$ from $1$ to $m$, choose the smallest $i \leq k$ such that $\phi_i > 0$ and set $\lambda_{ki}$ to $\min ($the amount needed to lower $\phi_i$ to $0, \lambda_{kk})$. If $\lambda_{kk} > 0$, choose the next smallest $i$ and repeat until $i = k$. We show that this algorithm produces a $\lambda$ which satisfies all three properties. First, $\phi_i$ for all $i$ begin as non-negative values and the algorithm never lowers $\phi_i$ below zero. Therefore, this $\lambda$ satisfies Property \ref{sprop1}. Second, at when $\lambda_{ki}$ is allocated, it must be the case that $\phi_i > 0$. In the previous iteration of our algorithm, when we were considering for which $j$ to increase $\lambda_{(k - 1)j}$, we chose the smallest $j$ such that $\phi_j > 0$. Therefore, it must be the case that $j \leq i$ (or we would have chosen to increase $\lambda_{(k - 1)i}$ instead of $\lambda_{(k - 1)j}$. This means that higher indices contribute their $\lambda$ to higher values, so Property \ref{sprop2} is satisfied. Third, note that $\lambda_{ki}$ is only positive if $\phi_{l} = 0$ for all $l < i$. Therefore Property $\ref{sprop3}$ is satisfied.
    
    Next, we need to show that any other $\lambda$ than the one we have proposed does not satisfy at least one of the three properties. Suppose there is some other $\lambda'$ which satisfies all three properties but differs from $\lambda$. Choose $i$ such that $i$ is the minimum index where $\lambda_{si}$ differs from $\lambda'_{si}$ for at least one $s$. Let $S$ be the set of all such $s$. Let $k$ be the minimum index in $S$. Let $\phi_i$ correspond to the result under $\lambda$ and $\phi'_i$ correspond to the result under $\lambda'$. We consider four cases:
    
    \textbf{Case 1: } Suppose $\phi_i > 0$ and $\phi_i > \phi'_i$. We show by contradiction there is no $\lambda'$ with these characteristics which satisfies all three properties. Because $\lambda$ satisfies Property \ref{sprop3} it must be the case that $\lambda_{sj} = 0$ for all $s \geq k, j > i$.  Then any decrease in $\phi_i$ must be the result of an increase in some $\lambda_{li}$ with $i < l < k$. However, this means that some $\lambda_{lt}$ for some $t < i$ must be decreased, because $\sum\limits_{u \leq l} \lambda_{lu} = 1$ and $\lambda_lu = 0$ for $u \geq i$. This contradicts our assumption that $i$ is the minimum index where $\lambda_{si}$ differs from $\lambda'_{si}$ because $\lambda_{st}$ differs from $\lambda'_{st}$ as well. \\
\\
    \textbf{Case 2: } We consider the case in which $\phi_i > 0$ and $\phi_i < \phi'_i$. Because $\phi_i < \phi'_i$, there must have been some $\lambda_{ki} > \lambda'_{ki}$. If $\lambda'_{ki}$ lower than $\lambda_{ki}$, there must be some $\lambda'_{kj} > \lambda_{kj}$ for $j \neq i$ in order to satisfy $\sum\limits_{u \leq k} \lambda_{ku} = 1$. By Property \ref{sprop3}, $\phi_j = 0$ for all $j < i$. Therefore there must be some positive $\lambda'_{kl}$ for $l > i$. However, this violates Property \ref{sprop3} because $\phi'_i > 0$. Therefore, this case cannot occur. \\ 
\\
    \textbf{Case 3: } Suppose $\phi_i = 0$ and $\lambda_{ki} < \lambda'_{ki}$. We show by contradiction there is no $\lambda'$ with these characteristics which satisfies all three properties. In order to satisfy $\sum\limits_{u \leq k} \lambda_{ku} = 1$, it must be true that $\lambda_{kj} > \lambda'_{kj}$ for some $j$. Furthermore, it must be the case that $j > i$ because otherwise $j$ would contradict our assumption that $i$ is the minimum index where $\lambda_{si}$ differs from $\lambda'_{si}$. Because of Property \ref{sprop1}, it must be the case that some $\lambda_{li} > \lambda'_{li}$ so that $\phi_i \geq 0$. We also know that $l > k$ because of the minimality of $k$. But then we have $i < j < k < l$ with $\lambda_{kj} > 0$ and $\lambda_{li} > 0$, which violates Property \ref{sprop2}. However, we know that $\lambda$ satisfies Property \ref{sprop2}, a contradiction.\\
\\
    \textbf{Case 4: } Suppose $\phi_i = 0$ and $\lambda_{ki} > \lambda'_{ki}$. We show by contradiction there is no $\lambda'$ with these characteristics which satisfies all three properties. In order to satisfy $\sum\limits_{u \leq k} \lambda_{ku} = 1$, it must be true that $\lambda_{kj} < \lambda'_{kj}$ for some $j$. Furthermore, it must be the case that $j > i$ because otherwise $j$ would contradict our assumption that $i$ is the minimum index where $\lambda_{si}$ differs from $\lambda'_{si}$. Because of Property \ref{sprop3} and $\lambda'_{kj} > 0$, we know that $\phi'_i = 0$. Therefore it must be the case that some $\lambda_{li} < \lambda'_{li}$. We also know that $l > k$ because of the minimality of $k$. But then we have $i < j < k < l$ with $\lambda'_{kj} > 0$ and $\lambda'_{li} > 0$, which violates Property \ref{sprop2}. Therefore, this case cannot occur.
    
    These cases are all-inclusive, so we have shown that the $\lambda$ generated by our algorithm is the unique $\lambda$ which satisfies all three properties.
\end{proof}

\end{document}